\documentclass[12pt]{elsarticle}

\usepackage{mathstyle00}
\usepackage{bm}
\usepackage{verbatim}
\usepackage{algorithm}   
\usepackage{algorithmic} 
\usepackage{amsmath}  
\usepackage{lineno}
\usepackage{graphicx}
\usepackage{epstopdf}
\usepackage{epsfig}
\usepackage{booktabs}
\usepackage{multirow}
\usepackage{tikz}

\usepackage{float}

\usepackage[margin=0.92in]{geometry}    
\newtheorem{thm}{Theorem}[section]
\newtheorem{cor}[thm]{Corollary}

\newtheorem{pro}[thm]{Proposition}

\begin{document}
\begin{frontmatter}

\title{A Bidirectional DeepParticle Method for Efficiently Solving Low-dimensional Transport Map Problems}

\author[hku]{Tan Zhang}
\ead{thta@connect.hku.hk}
\author[uoc]{Zhongjian Wang}
\ead{zhongjian.wang@ntu.edu.sg}
\author[uci]{Jack Xin}
\ead{jxin@math.uci.edu}
\author[hku,mil]{Zhiwen Zhang\corref{cor1}}
\ead{zhangzw@hku.hk}
\address[hku]{Department of Mathematics, The University of Hong Kong, Pokfulam Road, Hong Kong SAR, China.}
\address[uoc]{
Division of Mathematical Sciences, 
Nanyang Technological University, 
21 Nanyang Link, 637371, Singapore.}
\address[uci]{Department of Mathematics, University of California, Irvine, CA 92697, USA.}
\address[mil]{Materials Innovation Institute for Life Sciences and Energy (MILES), HKU-SIRI, Shenzhen, P.\,R. China.}
\cortext[cor1]{Corresponding author}

\begin{abstract}
\noindent This paper aims to efficiently compute transport maps between probability distributions arising from particle representation of bio-physical problems.
We develop a bidirectional DeepParticle (BDP) method to learn and generate solutions under varying physical parameters. Solutions are approximated as empirical measures of particles that adaptively align with the high-gradient regions. The core idea of the BDP method is to learn both forward and reverse mappings (between the uniform and a non-trivial target distribution) by minimizing the discrete \textit{2-Wasserstein distance} (W2) and optimizing the  transition map therein 
by a \textit{mini-batch} technique.
We present numerical results demonstrating the effectiveness of the BDP method for learning and generating solutions to Keller–Segel chemotaxis systems in the presence of laminar flows and Kolmogorov flows with chaotic streamlines in three space dimensions. The BDP outperforms two recent representative single-step flow matching and diffusion models (rectified flow and shortcut diffusion models) in the generative AI literature. However when the target distribution is high-dimensional (4 and above), e.g. a mixture of two Gaussians, the single-step diffusion models scale better in dimensions than BDP in terms of W2-accuracy. 

\end{abstract}
\begin{keyword}
Particle method; optimal transport; 
bidirectional mappings; 
deep neural networks; Keller–Segel system; 
one-step generation.

\end{keyword}

\end{frontmatter}

\section{Introduction}\label{sec:intro}
\noindent The evaluation of discrepancies between probability distributions represents a fundamental challenge in machine learning. For instance, generative models such as generative adversarial networks (GANs) and variational autoencoders (VAEs) \cite{GAN, kingma2013auto, dinh2017density} seek to transform data points into latent codes that conform to a basic (Gaussian) distribution, enabling the generation and manipulation of data. Representation learning is based on the premise that if a sufficiently smooth function can map a structured data distribution to a simple distribution, it is likely to carry meaningful semantic interpretations, which are advantageous for various downstream learning tasks. Conversely, domain transfer methods identify mappings to shift points between two distinct, empirically observed data distributions, targeting tasks such as image-to-image translation, style transfer, and domain adaptation \cite{zhu2017unpaired, courty2016optimal, trigila2016data, peyre2019computational}. These tasks can be formulated as finding a transport map between the two distributions:

\textit{Given two empirical samples/observations $X_0, X_1$ with $X_0 \sim \pi_0, X_1\sim \pi_1$ on the metric space $X$ and $Y$, find a transport map $f: X \to Y$, such that $f(X_0) \sim \pi_1$ when $X_0 \sim \pi_0$.} 

In recent years, flow models utilizing neural ordinary differential equations (ODEs) and score-based diffusion models with stochastic differential equations (SDEs) \cite{chen2018neural, papamakarios2021normalizing, song2020score} have been employed to address many problems in this domain. A numerical ODE/SDE solver is trained and used to simulate the inference process. These models typically introduce a virtual time
$t$ into the transformation between the two distributions, discretizing this virtual time to facilitate training and generation. This discrete process requires multiple calls to the neural network output, increasing the time needed for inference. Although these methods produce high-quality samples, they necessitate an iterative inference procedure, often involving dozens to hundreds of forward passes through the neural network, resulting in slow and costly generation.
To expedite the sampling process, it is essential to reduce the number of discrete steps required, allowing the entire sampling to be completed in just a few or even a single step. In doing so, these models sacrifice some performance of the original multi-step diffusion methods in favor of greater generation efficiency. Moreover, if the generation process is limited to a single step, these models can also be adapted to solve the optimal transport (OT) problem\cite{liu2022flow, frans2024one}.

The primary challenge lies in identifying suitable metrics that possess both good statistical and optimization properties for finding transport maps. The \textit{Wasserstein distance}, based on OT, has been employed for this purpose in various machine learning problems \cite{villani2021topics, ambrosio2021lectures, figalli2021invitation, peyre2019computational}. A particularly notable property of the \textit{Wasserstein distance} is its applicability between distributions that do not share the same support, which is often the case when working with empirical distributions.
In our previous works \cite{DP_22,wang2024deepparticle}, we developed a DeepParticle (DP) method to learn and generate solutions based on particle representation. For Keller–Segel (KS) chemotaxis systems \cite{keller1970initiation} that depend on physical parameters (e.g., flow amplitude in the advection-dominated regime and evolution time), DP minimizes the \textit{2-Wasserstein distance} between the source and target distributions. Unlike the multi-step iterative process in diffusion models, DP requires only a single call to a  neural network to approach the target distribution. Essentially, this constitutes a one-step process to map an initial distribution $\pi_0$ to a target distribution $\pi_1$ at given (physical) parameters.


In this work, we further develop an efficient deep learning approach, the Bidirectional DeepParticle (BDP) method, to learn and solve the physically parameterized transport map problems. On top of the uni-directional map in DP \cite{DP_22,wang2024deepparticle}, BDP will make the network learn both the forward mapping $f: X \to Y$ and the reverse mapping $g: Y \to X$ simultaneously to improve the performance and stability. Since the use of a costly transition matrix is unavoidable when calculating the \textit{2-Wasserstein distance}, we carry out a \textit{mini-batch} technique during the training process \cite{fatras2021minibatch, sommerfeld2019optimal}. We will introduce this technique and analyze its error in approximating the \textit{2-Wasserstein distance}. Additionally, we compare the BDP performance with two representative single-step diffusion (flow-matching) models in 
generating target distributions from Keller-Segel chemotaxis data, and mixture of Gaussians. In lower dimensions, such as $2$ and $3$ space dimensions in physics, or when data volume is not too large, we observed that DP models can achieve better accuracies than 
single-step diffusion and flow-matching models
while remaining efficient. 
However, the accuracies of DP models (in  \textit{Wasserstein distance}) decrease with increasing dimension beyond 3, while the single-step diffusion and flow-matching models are much less sensitive. This critical phenomenon may generalize to other data sets and is worth further study. For example, whether the critical dimension is almost universal or problem dependent.


The rest of the paper is organized as follows. In Section 2, we present our DeepParticle method with the \textit{bidirectional} idea to learn the forward and reverse transport maps between $\pi_0$ and $\pi_1$. In connection with OT, we briefly review the framework and algorithm of two single-step models, \textit{Rectified flow} \cite{liu2022flow} and \textit{Shortcut model} \cite{frans2024one}, for the subsequent comparison study. In Section 3, we show numerical results to demonstrate the performance of our method and compare with the single-step models.  Finally, concluding remarks and future work are in Section 4.

\section{Bidirectional Deep Particle Method}\label{sec:ConvergenceAnalysis}
\noindent
In this section, we would like to introduce our BDP method and its corresponding network architecture to learn the features of the transport map between two distributions $\pi_0$ and $\pi_1$. Compared with recent diffusion models (e.g. DDPM \cite{ho2020denoising}, DDIM \cite{song2020denoising}), the sampling process in Deep Particle methods can be viewed as being completed in only a single step. The mapping error of the Deep Particle methods is measured based on the \textit{2-Wasserstein distance} (W2).

\subsection{2-Wasserstein distance}
Given distributions $\pi_0$ and $\pi_1$ defined in the metric space $X$ and $Y$, we attempt to find a transport map $f_*^0: X \to Y$ such that $f_*^0(\pi_0) = \pi_1$, where $*$ denotes the push-forward of the transport map. For any function $f_*: X \to Y$, the \textit{2-Wasserstein distance} between $f_*(\pi_0)$ and $\pi_1$ can be defined by:
\begin{equation}
    W_2 (f_*(\pi_0), \pi_1) := \Big( \inf_{\gamma \in \Gamma (\pi_0, \pi_1)} \int_{X \times Y} c_1(f_*(x), y)^2 d\gamma(x, y) \Big)^{\frac{1}{2}} ,
    \label{W2forward}
\end{equation}
where $\Gamma(\pi_0, \pi_1)$ denotes the collection of all measures on $X \times Y$ with marginals $f_*(\pi_0)$ and $\pi_1$ on the first and second factors and $c_1$ denotes the metric (distance) on $Y$. Similarly, for any function $g_*: Y \to X$, we could define
\begin{equation}
    W_2(g_*(\pi_1), \pi_0) := \Big( \inf_{\gamma \in \Gamma (\pi_1, \pi_0)} \int_{Y \times X} c_0(g_*(y), x)^2 d\gamma(y, x) \Big)^{\frac{1}{2}} ,
\end{equation}
where $\Gamma(\pi_1, \pi_0)$ denotes the collection of all measures on $Y \times X$ with marginals $g_*(\pi_1)$ and $\pi_0$ on the first and second factors and $c_0$ denotes the metric (distance) on $X$. In practical implementation, the distribution $\pi_0$ and $\pi_1$ is approximated by the empirical distribution functions with the particles $N$ (samples), i.e. $\pi_0 = \frac{1}{N} \sum_{i=1}^N \delta_{x_i}, \pi_1 = \frac{1}{N} \sum_{i=1}^N \delta_{y_i}$. It is known that any joint distribution in $\Gamma(\pi_0, \pi_1)$ can be approximated by a $N \times N$ stochastic transition matrix \cite{sinkhorn1964relationship}, $\gamma = (\gamma_{ij})_{i,j}$ ,which satisfies 
\begin{equation}
    \forall i, j,\ \gamma_{ij} \geq 0;\ \forall i,\  \sum_{j=1}^N \gamma_{i,j} = 1;\ \forall j,\ \sum_{i=1}^N \gamma_{i,j} = 1.
\end{equation}
Then we can obtain the discretization of the \textit{2-Wasserstein distance} \eqref{W2forward}:
\begin{equation}
    \hat{W}_2(f) := \Big( \inf_{\gamma \in \Gamma^N} \frac{1}{N} \sum_{i,j}^N c_1(f(x_i), y_j)^2 \gamma_{i,j} \Big)^{\frac{1}{2}}.
    \label{disW2forward}
\end{equation}

\subsection{Methodology and network architecture}

Given the training datasets $\{x_i\}_{i=1}^M \subset \mathbb{R}^d$ and $\{y_j\}_{j=1}^M \subset \mathbb{R}^d$, we have derived \eqref{disW2forward} to be minimized using gradient descent. However, directly inputting all training data samples incurs significant memory costs for the transition matrix $\gamma \in \mathbb{R}^{M \times M}$. To mitigate this issue, we employ the \textit{mini-batch} technique commonly used in deep learning literature. Specifically, we select 
$N<M$ sub-samples from the data computed by the interacting particle method in each iteration of the training process, and we resample these sub-samples every 1000 iterations.

In solving problems such as KS systems \cite{keller1970initiation} that involve real physical parameters, we expect the network to effectively represent these parameters and learn how changes in them affect the target distribution. In this context, more than one set of training data ($\{x_i\}_{i=1}^N$ and $\{y_j\}_{j=1}^N$ consists of one set of data) should be assimilated. We denote the total number of distinct groups of physical parameters as $N_{dict}$. This means that the network will have $N_{dict}$ pairs of i.i.d. samples of input and output distribution, denoted by $\{x_{i,r}\}_{i=1}^N$ and $\{y_{j,r}\}_{j=1}^N$ for $r = 1 \cdots N_{dict}$. Correspondingly, we express these different physical parameters in terms of $\{\sigma_r\}_{r=1}^{N_{dict}}$ and let them be the inputs along with each set of $\{x_{i,r}\}_{i=1}^N$. This just means the input for the forward network is $\{(x_{i,r}, \sigma_r) \}_{i=1}^N$.

In addition, we expect this network to learn both the mapping from $\pi_0$ to $\pi_1$ and the mapping from $\pi_1$ to $\pi_0$, ensuring that these two mappings are consistent with each other. Consequently, there are effectively two sub-networks within the overall network architecture, and we include an error term in the loss function to verify this consistency. To be specific, for the network $f_{*}^{\theta}: X \to Y$ and $g_{*}^{\vartheta}: Y \to X$, the loss function can be represented by:
\begin{equation}
    Loss =  \hat{W}_2(f_{*}^{\theta} ) +  \hat{W}_2(g_{*}^{\vartheta}) + \lambda \cdot MSE(\bm x, g_{*}^{\vartheta} \circ f_{*}^{\theta} (\bm x; \sigma)),\ with
\end{equation}
\begin{equation}
    \hat{W}_2(f_{*}^{\theta} ) = \frac{1}{N \cdot N_{dict}} \sum_{r=1}^{N_{dict}} \Big( \inf_{\gamma_r \in \Gamma^N} \sum_{i,j = 1}^N |f^{\theta}_{*} (x_{i,r}; \sigma_r ) - y_{j,r}|^2 \gamma_{ij, r}  \Big), 
\end{equation}
\begin{equation}
    \hat{W}_2(g_{*}^{\vartheta} ) = \frac{1}{N \cdot N_{dict}} \sum_{r=1}^{N_{dict}} \Big( \inf_{\varphi_r \in \Gamma^N} \sum_{i,j = 1}^N |g^{\vartheta}_{*} (y_{i,r}; \sigma_r ) - x_{j,r}|^2 \varphi_{ij, r}  \Big), 
\end{equation}
\begin{equation}
    MSE(\bm x, g_{*}^{\vartheta} \circ f_{*}^{\theta} (\bm x; \sigma)) = \frac{1}{N \cdot N_{dict}} \sum_{r=1}^{N_{dict}} \sum_{i=1}^N ||x_{i,r} - g^{\vartheta}_{*} (f^{\theta}_{*} (x_{i,r}; \sigma_r ); \sigma_r )||^2,
    \label{MSEerr}
\end{equation}
where $\theta, \vartheta$ denote the parameters of two sub-networks respectively, $MSE(\cdot,\cdot)$ represents the mean square error between two groups of data, and $\lambda \in \mathbb{R}$ denotes its coefficient. Compared to the network architecture with only one forward mapping $f^{\theta}_{*}$, introducing the network $g^{\vartheta}_{*}$ can enhance the stability of the learned mapping (transition matrix). For example, we consider there exist two learned forward mapping networks $f^{\theta_1}_{*}$ and $f^{\theta_2}_{*}$, such that the difference between them is $f^{\theta_1}_{*}(x_1;\sigma) = f^{\theta_2}_{*}(x_2; \sigma) = y_1, f^{\theta_1}_{*}(x_2;\sigma) = f^{\theta_2}_{*}(x_1; \sigma) = y_2$. This situation may arise because we use the \textit{mini-batch} technique during the training process. Since the mapping between the two points is just swapped, the two networks will get the same result when calculating the \textit{2-Wasserstein distance}. For the model with two sub-networks, $f^{\theta}_{*}$ and $g^{\vartheta}_{*}$, since the training of $f^{\theta}_{*}$ and $g^{\vartheta}_{*}$ is independent and performed simultaneously, the two different forward networks $f^{\theta_1}_{*}$ and $f^{\theta_2}_{*}$ will have different performance in the \textit{MSE} loss term under a fixed $g^{\vartheta}_{*}$. For example, if the network $g^{\vartheta}_{*}$ has $g^{\vartheta}_{*}(y_1;\sigma) = x_3,\ g^{\vartheta}_{*}(y_2;\sigma) = x_4$, and $x_3 \neq x_4$, then it will occur the case that
\begin{equation}
    g^{\vartheta}_{*} (f^{\theta_1}_{*}(x_1;\sigma);\sigma) =  g^{\vartheta}_{*}(y_1;\sigma) = x_3,
\end{equation}
\begin{equation}
    g^{\vartheta}_{*} (f^{\theta_2}_{*}(x_1;\sigma);\sigma) =  g^{\vartheta}_{*}(y_2;\sigma) = x_4,
\end{equation}
and this will lead to different \textit{MSE} loss terms by following \eqref{MSEerr}. At this point, the network will tend to retain the one with the smaller error, thus avoiding the instability of the learned mapping caused by this.

For the architecture of the two sub-networks, we utilize fully connected networks (Multilayer Perceptrons). Each sub-network consists of three latent layers, with each layer having a width of 40 units. The activation function used is $tanh(\cdot)$. Then the relationship between two adjacent layers $l_i$ and $l_{i+1}$ can be represented by:
\begin{equation}
    l_{i+1} = tanh (W_i l_i + b_i),
\end{equation}
where $W_i$ is the weight matrix and $b_i$ is the bias vector of layer $l_i$. We perform similar operations for the output layer, but at this time we do not use the activation function.

\subsection{Mini-batch technique and relation to OT}

In this subsection, we aim to explain the rationale behind using the \textit{mini-batch} technique during the training process, as well as to analyze the errors induced by this approach. Directly computing the \textit{2-Wasserstein distance} between empirical probability distributions with $M$ points has a complexity of $O(M^3 \log M)$ \cite{peyre2019computational}, indicating that it is impractical for large data scenarios. To reduce this complexity, a promising technique is to regularize the Wasserstein distance with an entropic term. This enables the use of the efficient Sinkhorn-Knopp algorithm, which can be implemented in parallel and has a lower computational complexity of 
$O(M^2)$ \cite{altschuler2017near}. However, this complexity is still prohibitive for many large-scale applications.

To train a neural network on large-scale datasets using the \textit{2-Wasserstein distance}, several works have proposed leveraging a \textit{mini-batch} computation of OT distances and back-propagating the resulting gradient into the network. If we divide the data into $k$ batches, each with a batch size of $N$, this strategy results in a complexity of $O(kN^2)$, enabling the network to handle large datasets. However, the trade-off is that averaging several OT quantities between mini-batches introduces a deviation from the original OT problem. In the context of the OT problem with the entire dataset, the corresponding \textit{2-Wasserstein distance} can be described by

\begin{equation}
    \hat{W}_2(f(\bm x^{(M)}), \bm y^{(M)}) := \Big( \inf_{\gamma \in \Gamma^M} \frac{1}{M} \sum_{i,j}^M c_1(f(x_i), y_j)^2 \gamma_{i,j} \Big)^{\frac{1}{2}},
    \label{2WdisofM}
\end{equation}
where $f(\bm x^{(M)}) := \{ f(x_i)\}_{i=1}^M \subset \mathbb{R}^d$, and $\bm y^{(M)} := \{y_i\}_{i=1}^M \subset \mathbb{R}^d$. And in practice, we separate $\bm x^{(M)}$ into $k$ batches, $\bm x_1^{(N)}, \cdots, \bm x_k^{(N)}$ and each batch has $N$ data samples ($M=kN$). Similarly, we separate $\bm y^{(M)}$ and obtain $\bm y_1^{(N)}, \cdots, \bm y_k^{(N)}$. The averaged empirical optimal transport distance of the data with batch size $N$ can be represented by 
\begin{equation}
    \hat{W}_2^{(k)} (f) := \frac{1}{k}\sum_{h=1}^k \hat{W}_2 (f(\bm x_h^{(N)} ), \bm y_h^{(N)}).
    \label{2WdisofN}
\end{equation}

There have been some previous theoretical results about the general non-asymptotic guarantees for the quality of the approximation $\hat{W}_2^{(k)} (f)$ in terms of the expected $L_1$ error and $L_2$ error. Recall that, for $\alpha > 0$ the covering number $\mathcal{N} (X, \alpha)$ of $X$ is defined as the minimal number of closed balls with radius $\alpha$ and centers in $X$ that is needed to cover $X$. 
\begin{pro}[Theorem 2 in \cite{sommerfeld2019optimal}]
     \textit{Let $\hat{W}_2^{(k)}(f)$ be as in \eqref{2WdisofN} for any choice of $k \in \mathbb{N}^+$, then for any integer $q \geq 2$ and $l_{max} \in \mathbb{N}$:}
    \begin{equation}
        \mathbb{E} [|\hat{W}_2^{(k)} (f) - W_2(f(\pi_0), \pi_1)  |] \leq 2 \Psi_q^{\frac{1}{2}} N^{- \frac{1}{4}},
    \end{equation}
    \textit{where} $\Psi_q := 8 q^{4} (diam(X))^2 ( q^{-2(l_{max} + 1) } \sqrt{M} + \sum_{l=0}^{l_{max}}q^{-2l} \sqrt{\mathcal{N}(X, q^{-l} diam(X))} )$.
\end{pro}
 It can be observed that the expected $L_1$ error has a decay rate $O(N^{-\frac{1}{2p}})$ with respect to the batch size. And it has been shown that, in the Euclidean case, the optimal value for $q$ is $q = 2$. The mean square error also has an upper bound.
 \begin{pro}[Theorem 5 in \cite{sommerfeld2019optimal}]
    \textit{Let $\hat{W}_2^{(k)}(f)$ be as in \eqref{2WdisofN} for any choice of $k \in \mathbb{N}^+$. Then for any integer $q \geq 2$, the mean squared error of the empirical optimal transport distance can be bounded as} ($M=kN$)
    \begin{equation}
        \mathbb{E}[|\hat{W}^{(k)}_2 (f) - W_2 (f(\pi_0), \pi_1)|^2] \leq 18 \Psi_q N^{-\frac{1}{2}} = O(N^{-\frac{1}{2}}) 
    \end{equation}
    \label{prop2.2}
 \end{pro}

\begin{cor}
    \textit{Let $P$ denote the loss function used during the training process of the BDP method (as mentioned in Algorithm \ref{Alg2} line 14), then $P$ will converge to the loss function $P_0$, which is implemented without the mini-batch technique, on average in $L_2$ as the batch size $N$ increases and has a convergence rate $O(N^{-\frac{1}{2}})$, i.e.}
    \begin{align}
        &\mathbb{E}[|\frac{1}{k}\sum_{h=1}^k P_h - P_0|^2] = O(N^{-\frac{1}{2}}),\ where\\
        P_0 := \sum_{r = 1}^{N_{dict}} \big[&\sum_{i,j}^M \big(|f_{\theta} (x_{i,r}; \sigma_r) - y_{j,r}|^2 \gamma_{ij,r}^f + |g_{\vartheta} (y_{i,r}; \sigma_r) - x_{j,r}|^2 \gamma_{ij,r}^g\big) \notag \\ &+ \frac{\lambda}{M} \sum_{i=1}^M|g_{\vartheta}(f_{\theta}(x_{i,r};\sigma_r); \sigma_r) - x_{i,r}|^2\big].
    \end{align}
\end{cor}
\begin{proof}
    \begin{align}
        \frac{1}{k} \sum_{h=1}^k P_h  &= \sum_{r=1}^{N_{dict}} \big[\hat{W}_2^{(k)} (f) + \hat{W}_2^{(k)} (g) + \frac{\lambda}{k N} \sum_{h=1}^k \sum_{i=1}^N |g_{\vartheta}(f_{\theta}(x_{i,r}^{(h)},\sigma_r), \sigma_r) - x_{i,r}^{(h)}|^2 \big]\notag \\
        & = \sum_{r=1}^{N_{dict}}\big[\hat{W}_2^{(k)} (f) + \hat{W}_2^{(k)} (g) + \frac{\lambda}{M} \sum_{i=1}^M |g_{\vartheta}(f_{\theta}(x_{i,r},\sigma_r), \sigma_r) - x_{i,r}|^2 \big].
    \end{align}
    \begin{equation}
        P_0 = \sum_{r=1}^{N_{dict}}\big[\hat{W}_2(f(\bm x^{(M)}), \bm y^{(M)}) + \hat{W}_2(g(\bm y^{(M)}), \bm x^{(M)}) + \frac{\lambda}{M} \sum_{i=1}^M |g_{\vartheta}(f_{\theta}(x_{i,r};\sigma_r); \sigma_r) - x_{i,r}|^2 \big].
    \end{equation}
    Then we can obtain that 
    \begin{align}
        \mathbb{E}[|\frac{1}{k}&\sum_{h=1}^k P_h - P_0|^2] = \mathbb{E}[|\sum_{r=1}^{N_{dict}} [\hat{W}_2^{(k)} (f) + \hat{W}_2^{(k)} (g) - \hat{W}_2(f(\bm x^{(M)}), \bm y^{(M)}) - \hat{W}_2(g(\bm y^{(M)}), \bm x^{(M)})]|^2 ] \notag \\ 
        & \leq 2\mathbb{E} [\sum_{r=1}^{N_{dict}}[ |\hat{W}_2^{(k)} (f) - \hat{W}_2(f(\bm x^{(M)}), \bm y^{(M)})|^2 + |\hat{W}_2^{(k)} (g) - \hat{W}_2(g(\bm y^{(M)}), \bm x^{(M)})|^2]]\notag \\
        &= 2\sum_{r=1}^{N_{dict}} [\mathbb{E}[|\hat{W}_2^{(k)} (f) - \hat{W}_2(f(\bm x^{(M)}), \bm y^{(M)})|^2] + \mathbb{E}[|\hat{W}_2^{(k)} (g) - \hat{W}_2(g(\bm y^{(M)}), \bm x^{(M)})|^2]]\notag\\
        & \leq 4\sum_{r=1}^{N_{dict}} \Big[\mathbb{E}[|\hat{W}_2^{(k)} (f) - W_2(f(\pi_0), \pi_1)|^2] + \mathbb{E}[|W_2(f(\pi_0), \pi_1)-\hat{W}_2(f(\bm x^{(M)}), \bm y^{(M)}) |^2]\notag \\
        &\quad + \mathbb{E}[|\hat{W}_2^{(k)} (g) - W_2(g(\pi_1), \pi_0)|^2] + \mathbb{E}[|W_2(g(\pi_1), \pi_0)-\hat{W}_2(g(\bm y^{(M)}), \bm x^{(M)})\Big]\notag \\
        &\leq \sum_{r=1}^{N_{dict}} 72 (\Psi_{f,r} + \Psi_{g,r}) (N^{-\frac{1}{2}} + (kN)^{-\frac{1}{2}}) \leq \sum_{r=1}^{N_{dict}} 144 (\Psi_{f,r} + \Psi_{g,r}) N^{-\frac{1}{2}} = O(N^{-\frac{1}{2}}), 
    \end{align}
    where $\Psi_{f,r}, \Psi_{g,r}$ are constants and determined by Proposition \ref{prop2.2}.
\end{proof}

\begin{algorithm}[h]
    \caption{Mini-batch technique to get a statistical approximation of $W_2(f(\pi_0), \pi_1)$}
    \begin{algorithmic}[1]
        \STATE \textbf{Input:} Probability measures $\pi_0, \pi_1$, batch-size $N$ and number of batch $k$.
        \FOR{$i = 1 \cdots k$}
            \STATE Sample i.i.d $\{x_i\}_{i=1}^N \sim \pi_0, \{y_j\}_{j=1}^N \sim \pi_1$.
            \STATE Compute $\hat{W}_2^{(i)} \leftarrow \hat{W}_2(f(\bm x^{(N)}), \bm y^{(N)})$
        \ENDFOR
        \STATE \textbf{Return:} $\hat{W}_2^{(k)}(f) \leftarrow \frac{1}{k} \sum_{h=1}^k \hat{W}_2^{(h)}$.
    \end{algorithmic}
\end{algorithm}

\begin{algorithm}[h]
    \caption{BDP method}
    \begin{algorithmic}[1]
        \STATE \textbf{Input:} Randomly initialize weight parameters $f_{\theta}$ and $g_{\vartheta}$ in the network. Probability measures $\pi_0, \pi_1$, batch-size $N$ , number of physical parameters groups $N_{dict}$ and number of batch $k$.
        \FOR{$h = 1 \cdots k$}
            \FOR{$r = 1 \cdots N_{dict}$}
            \STATE Sample i.i.d $\{x_{i,r}\}_{i=1}^N \sim \pi_0, \{y_{j,r}\}_{j=1}^N \sim \pi_1$ w.r.t. physical parameter $\sigma_r$.
            \STATE Set $\gamma^{f}_{ij,r}, \gamma^{g}_{ij,r} \leftarrow \delta_{ij}$, i.e. initialize permutation matrices;
        \ENDFOR    
        \IF{not the first training mini-batch}
        \FOR{$r = 1 \cdots N_{dict}$}
        \STATE Solve the Earth Movers distance problem between $f_{\theta}(\bm x_r)$ and $\bm y_r$, and update the permutation matrix $\bm \gamma^f_r$ i.e. $\bm \gamma^f \leftarrow ot.emd(\bm a,\bm b, ot.dist(f_{\theta}(\bm x_r,\sigma_r)), \bm y_r)$, where $\bm a= \bm b = (\frac{1}{N}, \cdots, \frac{1}{N})^T \in \mathbb{R}^N$ and using the OT package
        \STATE Solve the Earth Movers distance problem between $g_{\vartheta}(\bm y_r,\sigma_r)$ and $\bm x_r$, and update the permutation matrix $\bm \gamma^g_r$
        \ENDFOR
        \ENDIF
        \REPEAT
        \STATE Compute the loss $P = \sum_{r = 1}^{N_{dict}} \sum_{i,j}^N \big(|f_{\theta} (x_{i,r}, \sigma_r) - y_{j,r}|^2 \gamma_{ij,r}^f + |g_{\vartheta} (y_{i,r}, \sigma_r) - x_{j,r}|^2 \gamma_{ij,r}^g\big) + \frac{\lambda}{N} \sum_{i=1}^N|g_{\vartheta}(f_{\theta}(x_{i,r},\sigma_r), \sigma_r) - x_{i,r}|^2$
        \STATE $\theta \leftarrow \theta - \delta_1\nabla_{\theta} P$,where $\delta_1$ is the learning rate
        
        \STATE $\vartheta \leftarrow \vartheta - \delta_2\nabla_{\vartheta} P$, where $\delta_2$ is the learning rate
        \STATE Repeat the process (8)-(11) and update the permutation matrices
        \UNTIL{given steps for each training mini-batch;}
        \ENDFOR
        \STATE \textbf{Return:} Save the trained model.
    \end{algorithmic}
    \label{Alg2}
\end{algorithm}

\subsection{Brief introduction to related approaches}

\subsubsection{Rectified flow}
The rectified flow \cite{liu2022flow} is an ODE model that transports distribution $\pi_0$ to $\pi_1$ by following straight line paths as much as possible. Straight paths are preferred both theoretically, as they represent the shortest distance between two endpoints, and computationally because they can be simulated exactly without the need for time discretization. Consequently, flows along straight paths bridge the gap between one-step and continuous-time models.
\begin{equation}
    dZ_t = v(Z_t, t) dt
\end{equation}

The drift $v$ drives the flow to follow the direction $(X_1 - X_0)$ of the linear path pointing from $X_0$ to $X_1$ as much as possible, by solving a simple least squares regression problem:
\begin{equation}
    \min_{v} \int_{0}^{1} \mathbb{E} [||(X_1 - X_0) - v(X_t, t) ||^2] dt \quad with\ X_t = t X_1 + (1- t) X_0
    \label{LSEofRectFlow}
\end{equation}

The main algorithm of rectified flow can be divided into the following steps:
\begin{itemize}
    \item Input: Draw $(X_0, X_1)$ from $\pi_0$ and $\pi_1$ and initialize $v^{\theta}$ with parameter $\theta$.
    \item Training: Solve the least squares regression problem \eqref{LSEofRectFlow} with $t \sim U(0,1)$ and get $v^{\hat{\theta}}$.
    \item Sampling: Draw $(Z_0, Z_1)$ following $dZ_t = v^{\hat{\theta}}(Z_t, t)dt$ starting from $Z_0 \sim \pi_0$. 
    \item Reflow: Start from $(Z_0, Z_1) = (X_0, X_1)$ and repeat the previous three steps.
\end{itemize}
\subsubsection{Shortcut diffusion model}
Shortcut diffusion model \cite{frans2024one} is a family of generative models that use a single network and training phase to produce high-quality samples in single or multiple sampling steps. Unlike distillation or consistency models, shortcut models are trained in a single training run without a schedule.
It defines $X_t$ as a linear interpolation between a data point $X_1 \sim \pi_1$ and a noise point $X_0 \sim N(\bm 0, \mathbb{I})$. The velocity field $v_t$ is the direction from the noise to the data point, i.e.
\begin{equation}
    X_t = (1-t) X_0 + t X_1,\ and\ v_t = X_1 - X_0. 
\end{equation}
Given only $X_0$ renders $v_t$ a random variable because there are multiple plausible pairs $(X_0, X_1)$ and different values that the velocity can take on.  The shortcut diffusion model would like to train a single model that supports different sampling budgets, by conditioning the model not only on the timestep $t$ but also on a desired step size $d$. Conditioning on $d$ allows shortcut models to account for the future curvature, and jump to the correct next point rather than going off track. They refer to the normalized direction from $X_t$ towards the correct next point $X'_{t+d}$ as the shortcut $s(X_t, t, d)$, i.e. $X'_{t+d} = X_t + s(X_t, t, d)$.

By leveraging an inherent self-consistency property of shortcut models, namely that one shortcut step equals two consecutive shortcut steps of half the size, i.e.
\begin{equation}
    s(X_t, t, 2d) = s(X_t, t, d) / 2 +  s(X'_{t+d}, t, d) / 2.
\end{equation}
In practice, they approximate $s(X_t, t, d)$ by $s^{\theta}(X_t, t, d)$ which is learned by the neural network. During the training process, they split the batch into a fraction that is trained with $d = 0$ and another fraction with randomly sampled $d > 0$ targets and arrive the combined loss function as
\begin{equation}
    \mathcal{L}^S(\theta) = \mathbb{E}_{(t,d) \sim p(t,d)} [||s^{\theta}(X_t, t, 0) - (X_1 - X_0) ||^2 + ||s^{\theta}(X_t, t, 2d) - s_{target}||^2],
    \label{lossfctofSC}
\end{equation}
where $ s_{target} =  s^{\theta}(X_t, t, d) / 2 +  s^{\theta}(X'_{t+d}, t, d) / 2$. In the practical implementation, $d$ takes discrete values and has a minimum unit $\frac{1}{M}$ and the main algorithm can be divided into the following steps:
\begin{itemize}
    \item Input: Initialize $X_0 \sim N(\bm 0,\mathbb{I})$, $X_1 \sim \pi_1$, and $(t,d) \sim p(t,d)$. 
    \item Training: Learn and obtain $\theta$ by minimize the loss function \eqref{lossfctofSC} (For the first $k$ batch, set $d = 0, s_{target} = X_1 - X_0$).
    \item Sampling: Set $X\sim N(\bm 0,\mathbb{I})$, $d = \frac{1}{M}$, $t = 0$. For $n \in [0, \cdots, M-1]$, compute $X \leftarrow X + s^{\theta}(X, t+nd, d)$.
\end{itemize}

\section{Numerical Experiments}\label{sec:NumericalResults}
\noindent
In this section, we will present several numerical examples computed by our BDP method and compare them with several one-step models. 
\subsection{KS simulation and generation in the presence of 3D Laminar flow}

The KS model, a partial differential equation system describing chemotaxis-driven aggregation in Dictyostelium discoideum \cite{keller1970initiation}. 
A common form of the KS model can be represented by: 
\begin{equation}
    \rho_t = \nabla \cdot (\mu \nabla\rho - \chi \rho \nabla c),\quad  \epsilon c_t = \Delta c - k^2 c + \rho,
    \label{KSEqu}
\end{equation}
where $\rho$ is the density of the bacteria, $c$ is the concentration of the chemo-attractant, and $\mu, \chi, \epsilon, k$ are non-negative constants. The coupled processes in \eqref{KSEqu} capture the motion of bacteria that diffuse with mobility $\mu$ and drift in the direction of $\nabla c$ with velocity $\chi \nabla c$. 

When parameters $(\epsilon, k) = \bm 0$, the equation related to the concentration $c$ in \eqref{KSEqu} is reduced to a simple Poisson equation $\rho = - \Delta c$. When proper boundary conditions are imposed on $c$, the classical solution assumes the convolution form $c = - \mathcal{K} * \rho$ with $\mathcal{K}$ representing the Green's function of the Laplacian. By substituting this solution into the equation related to the density $\rho$ in \eqref{KSEqu}, we can obtain that
\begin{equation}
    \rho_t = \mu \Delta \rho + \chi \nabla \cdot (\rho \nabla(\mathcal{K * \rho})). 
    \label{KSmodellinrho}
\end{equation}
The original KS system is reduced to a non-local non-linear advection-diffusion PDE. To capture chemo-tactic transport phenomena in fluid dynamic settings such as marine hydrodynamic, previous investigations have focused on the modified transport equation
\begin{equation}
    \mathcal{L}_{\bm v} \rho = \mu \Delta \rho + \chi \nabla \cdot (\rho \nabla(\mathcal{K * \rho})),
    \label{MLKSequ}
\end{equation}
where $\mathcal{L}_{\bm v} := \partial_t \rho + \nabla\cdot(\bm v\rho)$ represents the advective Lie derivative accounting for  fluid velocity field $\bm v$. Here we consider the KS model with advection term and $\bm v$ is a divergence-free velocity field. Under the influence of the environmental velocity field $\bm v$, the movement of the organism and the blow-up behavior of the model are also affected and deserve to be investigated. Equation \eqref{MLKSequ} can be approximated by the interacting particle system below:
\begin{equation}
    d \bm X_j = -\frac{\chi M}{J} \nabla_{\bm X_j} \sum_{i=1, i\neq j}^{J} \mathcal{K}(|\bm X_i - \bm X_j|) \, dt + \bm v(\bm X_j)\, dt +  \sqrt{2\mu} \, d \bm W_j,\ j = 1,\cdots, J,
    \label{IPMofKS}
\end{equation}
with $M$ is the conserved total mass and $\{\bm W_j\}_{j = 1}^J$ are independent $d$-dimensional Brownian motions. 
In practice, numerical instability emerges in the chemo-attractant component $\sum_{i=1, i\neq j}^{J}$ $ \mathcal{K}(|\bm X_i - \bm X_j|)$ as particles concentrate, due to unbounded kernel contributions at small interparticle distances. To avoid this situation, we replace the original kernel $\mathcal{K}(\cdot)$ with a smoothed approximation $\mathcal{K}_{\delta}(\cdot)$, such that $\mathcal{K}_{\delta} (z) \to \mathcal{K}(z)$ as $\delta \to 0$, and $\delta$ here is a regularization parameter. For example, we can define $\mathcal{K}_{\delta} (z) = \mathcal{K} (z) \cdot \frac{|z|^2} { |z|^2 + \delta^2}$ and obtain the following regularized SDE system:
\begin{equation}
     d \bm X_j = -\frac{\chi M}{J} \nabla_{\bm X_j} \sum_{i=1, i\neq j}^{J} \mathcal{K}_{\delta}(|\bm X_i - \bm X_j|) \, dt + \bm v(\bm X_j)\, dt +  \sqrt{2\mu}\,  d \bm W_j,\ j = 1,\cdots, J.
    \label{IPMofKSreg}
\end{equation}
A well-studied KS model is a 2-dimensional system described by \eqref{KSmodellinrho} with $\mu = \chi = 1$. The total mass of this system satisfies the conservation law, i.e.
\begin{equation}
    \frac{d}{dt} \int_{\mathbb{R}^2} \rho(\bm x, t) d\bm x = 0.
\end{equation}
This is also true for the system described by \eqref{MLKSequ} if the velocity field $\bm v$ is divergence-free. If we denote $M :=  \int_{\mathbb{R}^2} \rho(\bm x, t) d\bm x$, then the second moment will have a fixed derivative with respect to time $t$, i.e.
\begin{equation}
    \frac{d}{dt} \int_{\mathbb{R}^2} |x|^2 \rho(x, t) dx = \frac{M}{2\pi} (8\pi - M),
\end{equation}
where $8\pi$ is called the critical mass of the KS system. In this regard, it is well recognized that if $M = 8\pi$, then the system will have a global smooth solution which will blow up as $t \to \infty$. It has been proved that if the total mass is smaller than the critical mass, the system \eqref{KSmodellinrho} and the system \eqref{MLKSequ} will have a global smooth solution with smooth initial data. In situations involving supercritical mass,  numerical experiments indicate that if advection is sufficiently large, it can prevent the solutions from blowing up.
We now consider three space dimensions (3D) and divergence-free velocity field to be a laminar flow, i.e.
\begin{equation}
    \bm v (x, y, z) = \sigma \cdot \Big(\exp(-y^2 - z^2),\ 0,\ 0 \Big)^T,
\end{equation}
which describes a flow of amplitude $\sigma$ traveling along the $x$-direction with speed depending on the radial position in the $yz$-plane. The training data is generated by the regularized interacting particle method \eqref{IPMofKSreg}. We first let the network learn the dependence of the aggregation patterns on the amplitude of the advection field $\sigma$ while fixing the evolution time $T = 0.02$. The physical parameter $\sigma$ of the data samples sent for training is isometrically distributed between $\sigma = 10$ and $\sigma = 150$. Since we are interested in the performance of small size models so that they can be implemented very efficiently, we set the network parameter size to 3k in subsequent experiments. Before comparing the performance of various models, we investigate the dependence of the new BDP model on $\lambda$, the coefficient of the identical verification error term. 

\begin{table}[H]
    \centering
    \caption{BDP Performance in 3D Laminar flow as $\lambda$ varies (best in bold).}
    \begin{tabular}{ccccccc}
        \toprule
        \multirow{2}{*}{Sigma} & &
        & \multicolumn{2}{c}{$W_2$ distance } \\
        \cline{2-7}
         & $\lambda = 0$ & $\lambda = 10^{-5}$ & $\lambda = 10^{-4}$ &$\lambda = 10^{-3}$  & $\lambda = 10^{-2}$ &  $\lambda = 10^{-1}$       \\
        
        \midrule
        $\sigma = 10^{(\circ)}$   & 0.0043 & 0.0034 & 0.0028 & \textbf{0.0027} &  0.0057   &  0.0096  \\
        $\sigma = 30^{(\blacktriangle)}$      & 0.0046 & 0.0061 & 0.0049 & \textbf{0.0025} &  0.0031  &  0.0109\\
        $\sigma = 50^{(\blacktriangle)}$    & 0.0062 & 0.0085 & 0.0069 & \textbf{0.0061} &  0.0063  & 0.0115 \\
        $\sigma = 80^{(\blacktriangle)}$    & 0.0127 & 0.0129 & 0.0151 & \textbf{0.0112} &  0.0147  & 0.0171 \\
        $\sigma = 100^{(\blacktriangle)}$    & 0.0059 & 0.0135 & 0.0101 & 0.0079 &  \textbf{0.0051}  & 0.0074  \\
        $\sigma = 120^{(\blacktriangle)}$     & 0.0078 & 0.0194 & 0.0114 & 0.0088 &  \textbf{0.0062}  & 0.0102 \\
        $\sigma = 150^{(\circ)}$      & 0.0136 & 0.0270 & 0.0151 & 0.0171 &  \textbf{0.0049}  & 0.0223 \\
        $\sigma = 200^{(\bullet)}$    & \textbf{0.2734} & 0.4093 & 0.3513 & 0.3788 &  0.4366  & 0.5534\\
        
        \bottomrule
    \end{tabular}
    \label{DPBirTable}
\end{table}
 In Table \ref{DPBirTable} and subsequent tables, we use different superscript symbols to indicate different types of parameters.  Here, $\circ, \blacktriangle, \bullet$ represent the data samples that are used (generated) in \textit{training}, \textit{interpolation}, and \textit{extrapolation} respectively under the corresponding physical parameters.
 From Table \ref{DPBirTable}, it can be stated that for the majority of cases, the value of \textit{2-Wasserstein distance} is minimized when $\lambda = 10^{-3}$ or $\lambda = 10^{-2}$, and the corresponding models achieve optimal performance. Therefore, in the following experiments, we test the performance of the BDP model under both these two lambda and choose the one that performed better. 

\begin{table}[H]
    \centering
    \caption{Model Comparison in 3D Laminar flow (best in bold).}
    \begin{tabular}{cccccc}
        \toprule
        \multirow{3}{*}{Sigma $\setminus$ Model} & &
        \multicolumn{3}{c}{$W_2$ distance (network parameter size $3k$)}  \\
        \cline{2-6}
          &  DM & DM & DM & DP & DP\\
        \cline{2-6}
          &   & Rectified flow & Shortcut  & & Bi-direction\\
        
        \midrule
        $\sigma = 10^{(\circ)}$   & 0.0137 & 0.0202 & 0.0184 & 0.0044 &  \textbf{0.0027} \\
        $\sigma = 30^{(\blacktriangle)}$    & 0.0158 & 0.0250 & 0.0187 & 0.0046 &  \textbf{0.0025}    \\
        $\sigma = 50^{(\blacktriangle)}$    & 0.0208 & 0.0218 & 0.0232 & 0.0064 &  \textbf{0.0061}    \\
        $\sigma = 80^{(\blacktriangle)}$    & 0.0174 & 0.0285 & 0.0217 & 0.0127 &  \textbf{0.0112}    \\
        $\sigma = 100^{(\blacktriangle)}$   & 0.0204 & 0.0284 & 0.0235 & 0.0137 &  \textbf{0.0051}   \\
        $\sigma = 120^{(\blacktriangle)}$   & 0.0275 & 0.0312 & 0.0390 & 0.0078 &  \textbf{0.0062}    \\
        $\sigma = 150^{(\circ)}$    & 0.0398 & 0.0522 & 0.0455 & 0.0249 &  \textbf{0.0049}   \\
        $\sigma = 200^{(\bullet)}$   & 0.2859 & 0.3121 & 0.3429 & \textbf{0.2580} & 0.2734\\
        
        \bottomrule
    \end{tabular}
    \label{laminarTable}
\end{table}

        
        

Table \ref{laminarTable} shows the \textit{2-Wasserstein distance} between different models and the reference distribution generated by the interacting particle method. In this table, we use DM to denote diffusion/flow matching model, and DP to denote our deep particle method. The same parameter size (\textit{3k}) was used for all methods. In the third line, \textit{Rectified flow} and \textit{Shortcut} means two models with single-step generation techniques which have been mentioned above.  \textit{Bi-direction} refers to BDP method which learns both  forward and reverse mappings.  It is seen that the two Deep Particle methods perform better than other single step models. The two Deep Particle methods both perform well and closely. The inference times of all the models are listed in the following table. The 
differences among one-step methods 
(one call of network during inference) are minor (i.e. of the same order) mainly due to programming languages for implementation with (DM, Rectified flow and shortcut DM) in Pytorch, 
and (DP, BDP) in JAX.


\begin{table}[H]
    \centering
    \begin{tabular}{ccccccc}
    \toprule
        Model &  & Diffusion Model&  Rectified flow & Shortcut& DeepParticle  & BDP \\
        Inference time & & 49 s & 0.48 s & 0.54 s & 0.24 s & 0.26 s\\
    \bottomrule
    \end{tabular}
    \caption{Inference times in second (s) of DM and one-step models.}
    \label{InferTime}
\end{table}

\begin{figure}[H]
    \centering
    \begin{subfigure}{0.31\textwidth}
         \includegraphics[width = \linewidth, height = 4cm]{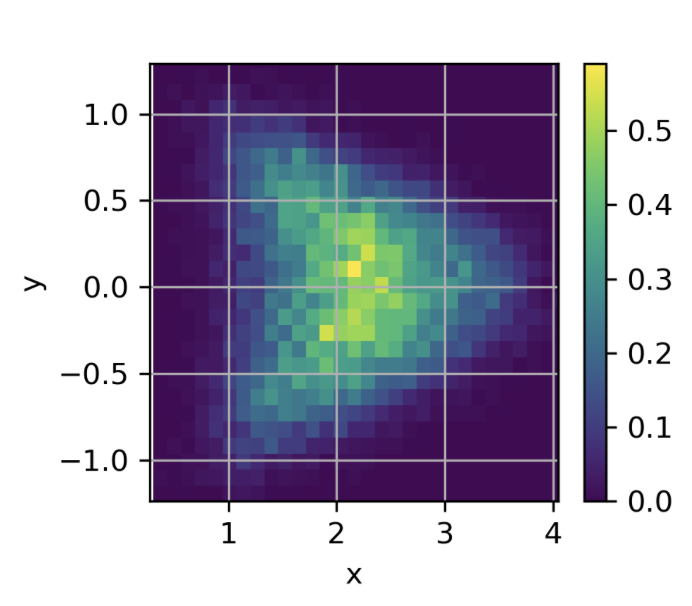}   
    \end{subfigure}
    \begin{subfigure}{0.31\textwidth}
         \includegraphics[width = \linewidth, height = 4cm]{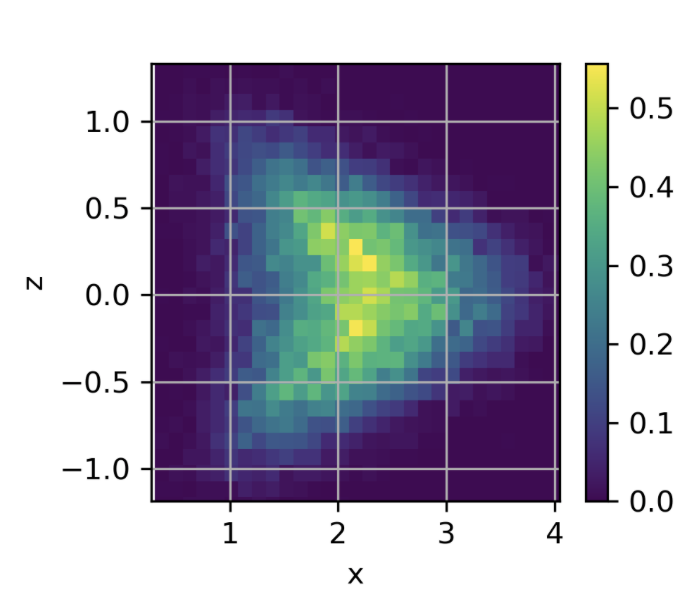}    
    \end{subfigure}
    \begin{subfigure}{0.31\textwidth}
         \includegraphics[width = \linewidth, height = 4cm]{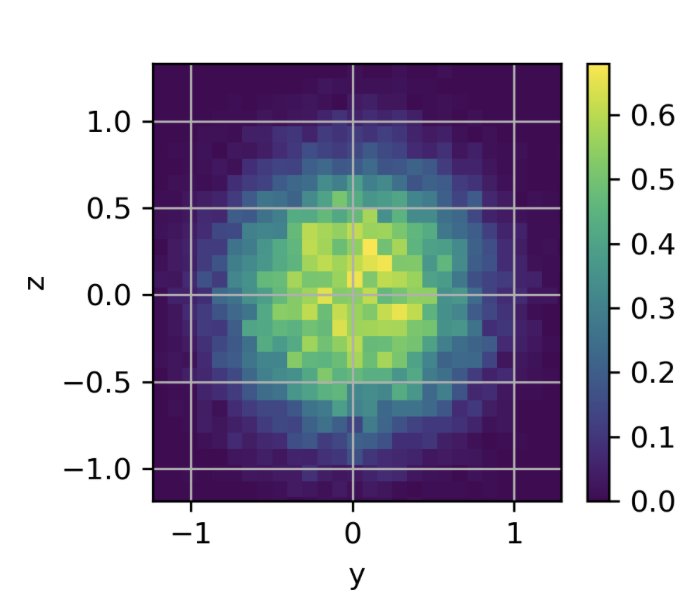}   
    \end{subfigure}
    \caption{Cross sectional views of  reference in 3D Laminar flow at $\sigma = 150$.}
    \label{Ref3DLaminar150}
\end{figure}

\begin{figure}[H]
    \centering
    \begin{subfigure}{0.31\textwidth}
         \includegraphics[width = \linewidth, height = 4cm]{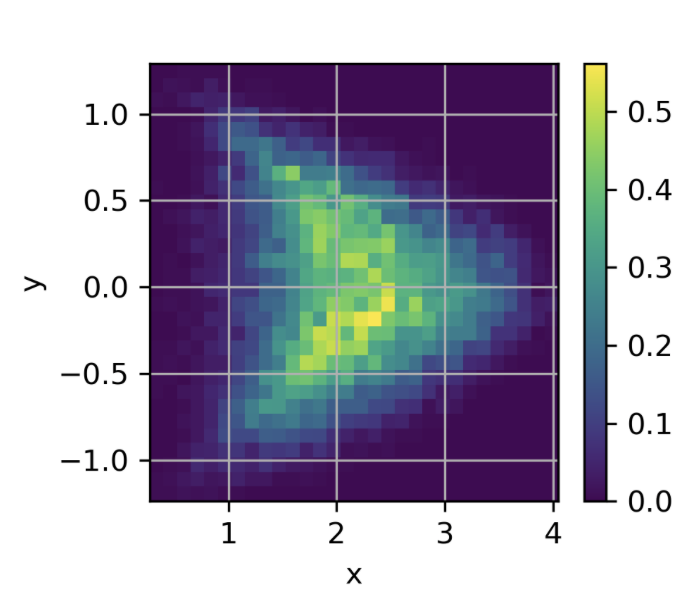}   
    \end{subfigure}
    \begin{subfigure}{0.31\textwidth}
         \includegraphics[width = \linewidth, height = 4cm]{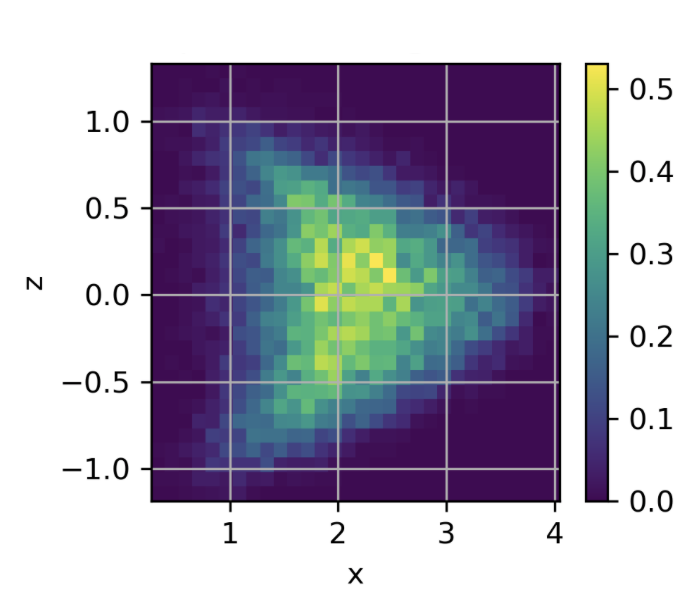}    
    \end{subfigure}
    \begin{subfigure}{0.31\textwidth}
         \includegraphics[width = \linewidth, height = 4cm]{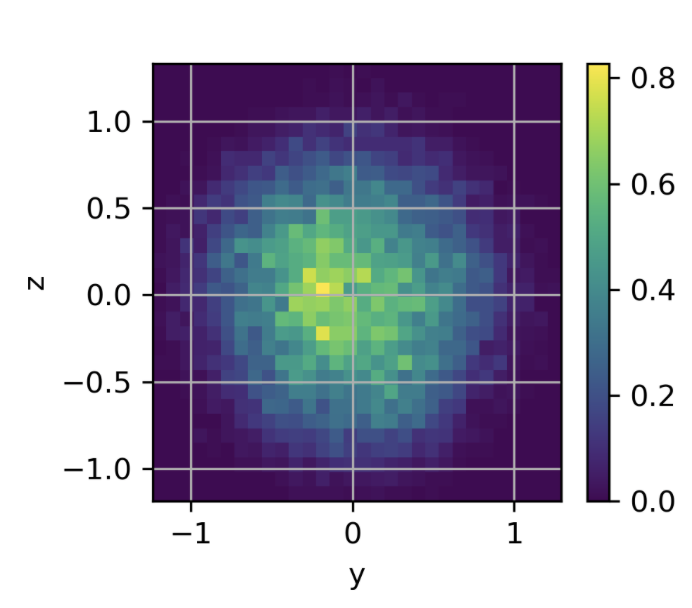}   
    \end{subfigure}
    \caption{Cross sectional views of DP method in 3D Laminar flow at $\sigma = 150$.}
    \label{DP3DLaminar150}
\end{figure}

\begin{figure}[H]
    \centering
    \begin{subfigure}{0.31\textwidth}
         \includegraphics[width = \linewidth, height = 4cm]{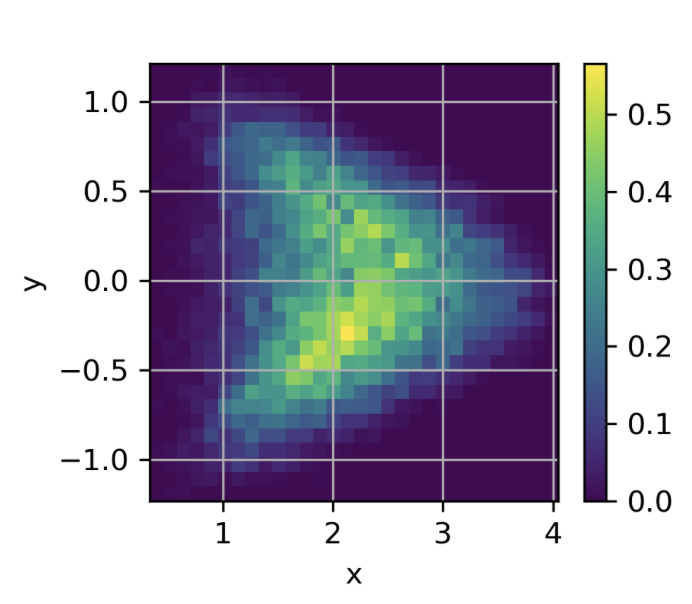}   
    \end{subfigure}
    \begin{subfigure}{0.31\textwidth}
         \includegraphics[width = \linewidth, height = 4cm]{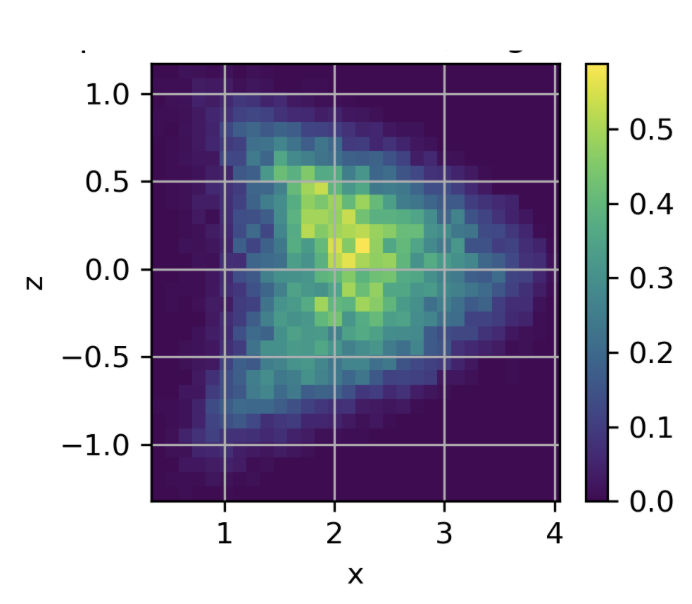}    
    \end{subfigure}
    \begin{subfigure}{0.31\textwidth}
         \includegraphics[width = \linewidth, height = 4cm]{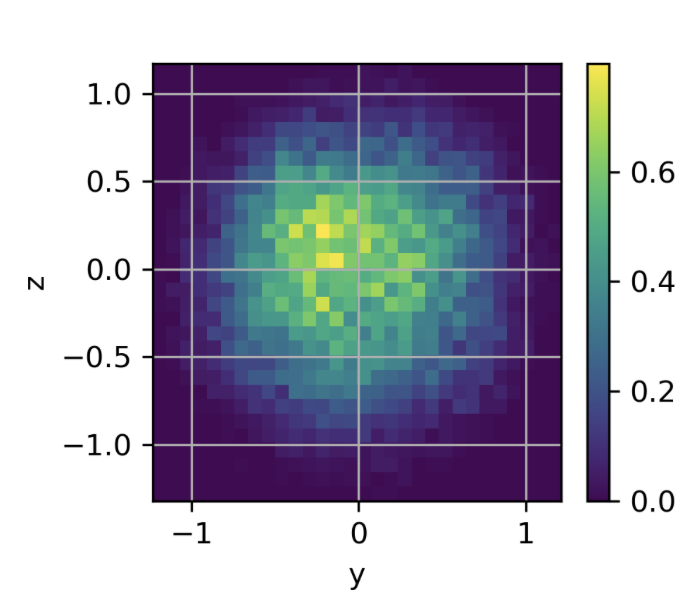}   
    \end{subfigure}
    \caption{Cross sectional views of BDP method in 3D Laminar flow at $\sigma = 150$.}
    \label{DPBir3DLaminar150}
\end{figure}

\begin{figure}[H]
    \centering    
    \begin{subfigure}{0.31\textwidth}
         \includegraphics[width = \linewidth, height = 4cm]{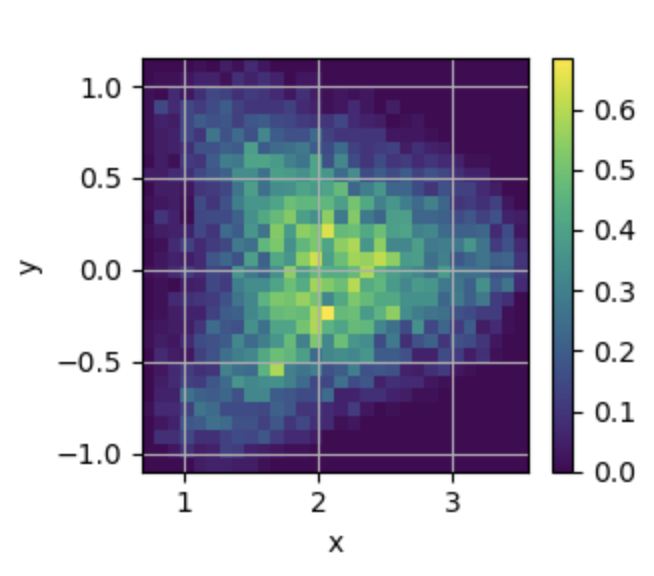}   
    \end{subfigure}
    \begin{subfigure}{0.31\textwidth}
         \includegraphics[width = \linewidth, height = 4cm]{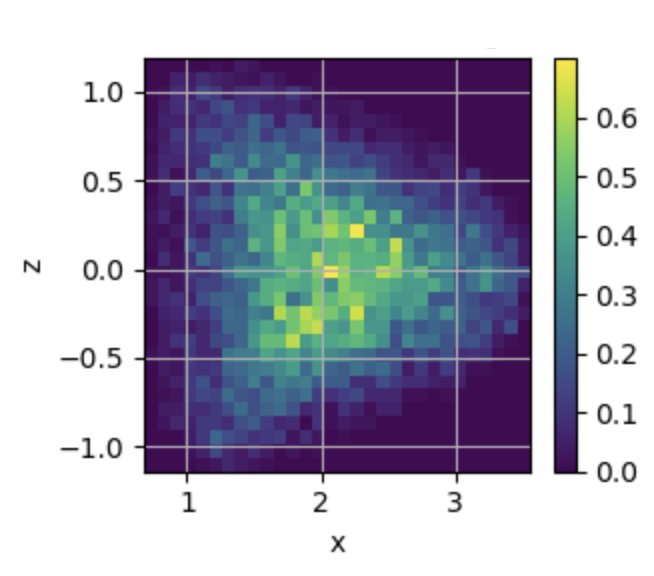}    
    \end{subfigure}
    \begin{subfigure}{0.31\textwidth}
         \includegraphics[width = \linewidth, height = 4cm]{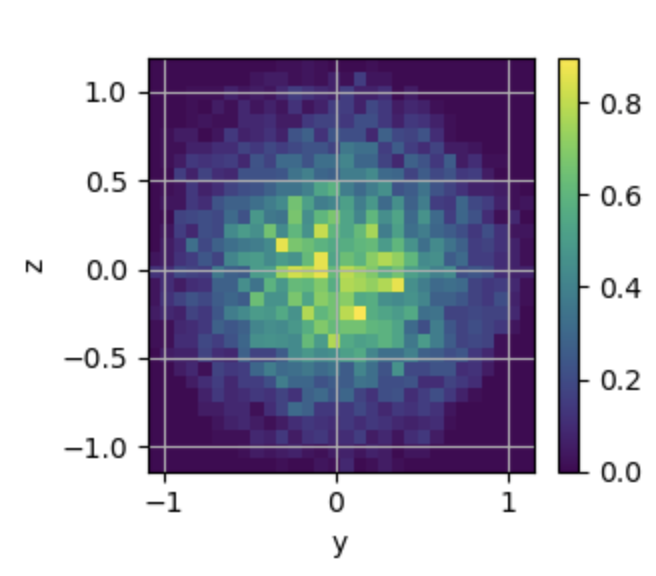}   
    \end{subfigure}
    \caption{Cross sectional views of Rectified flow in 3D Laminar flow at $\sigma = 150$.}
    \label{DMRect3DLaminar150}
\end{figure}

\begin{figure}[H]
    \centering
    \begin{subfigure}{0.31\textwidth}
         \includegraphics[width = \linewidth, height = 4cm]{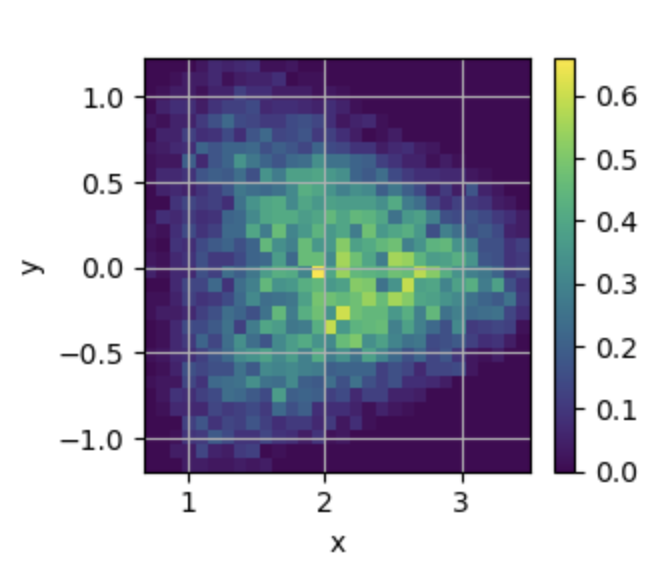}   
    \end{subfigure}
    \begin{subfigure}{0.31\textwidth}
         \includegraphics[width = \linewidth, height = 4cm]{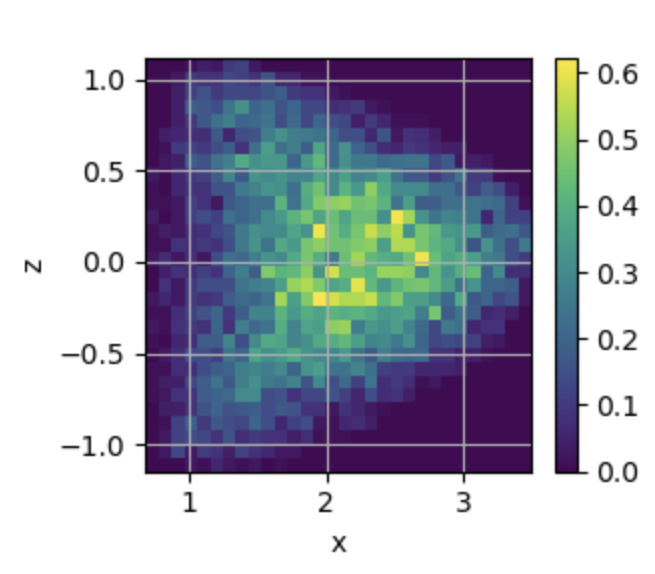}    
    \end{subfigure}
    \begin{subfigure}{0.31\textwidth}
         \includegraphics[width = \linewidth, height = 4cm]{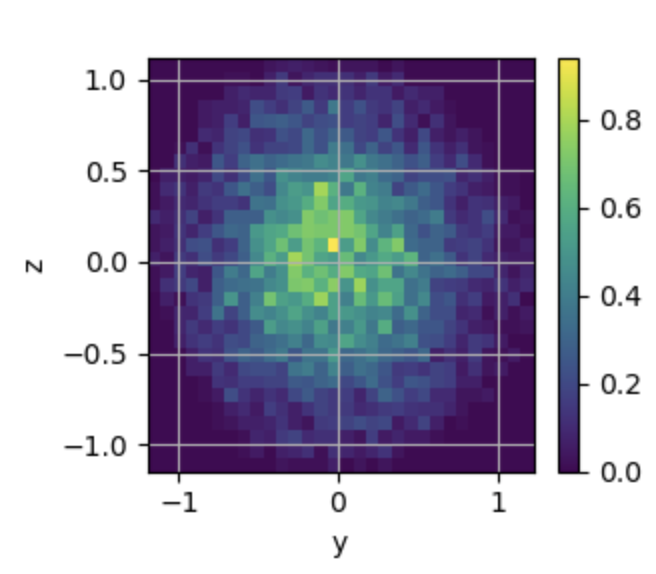}   
    \end{subfigure}
    \caption{Cross sectional views of Shortcut model in 3D Laminar flow at $\sigma = 150$.}
    \label{DMShort3DLaminar150}
\end{figure}

Figure \ref{Ref3DLaminar150} to \ref{DMShort3DLaminar150} show the distributions generated by different models and the reference in the case of $\sigma = 150$ in 3D laminar flow. It can be seen that other single-step models do not learn the tail well compared to DP models.

\subsection{KS in 3D Kolmogorov flow}
We consider KS in 3D Kolmogorov flow to generate chemotaxis patterns while the organisms travel and aggregate in chaotic streamlines. The fluid velocity field is:
\begin{equation}
    \bm v (x, y, z) = \sigma \cdot \Big(\sin(2\pi z),\ \sin(2\pi x),\ \sin(2\pi y) \Big)^T.
\end{equation}

The rest of the experimental setup is the same as in the previous subsection on 3D laminar flow. We similarly compare the performance of different models at small network parameter sizes.

\begin{table}[h]
    \centering
    \caption{Model Performance in 3D Kolmogorov flow}
    \begin{tabular}{cccccc}
        \toprule
        \multirow{3}{*}{Sigma $\setminus$ Model} & 
        & \multicolumn{3}{c}{$W_2$ distance (network parameter size $3k$)} \\
        \cline{2-6}
          & DM & DM & DM & DP & DP\\
        \cline{2-6}
           &   & Rectified flow & Shortcut  & & Bi-direction\\
        
        \midrule
        $\sigma = 10^{(\circ)}$    & 0.0184 & 0.0219 & 0.0221 & 0.0113 &  \textbf{0.0066}       \\
        $\sigma = 30^{(\blacktriangle)}$     & 0.0140 & 0.0196 & 0.0182 & 0.0064 &  \textbf{0.0058}    \\
        $\sigma = 50^{(\blacktriangle)}$     & 0.0297 & 0.0417 & 0.0343 & \textbf{0.0060} &  0.0127    \\
        $\sigma = 80^{(\blacktriangle)}$     & 0.0311 & 0.0471 & 0.0397 & 0.0167 &  \textbf{0.0153}    \\
        $\sigma = 100^{(\blacktriangle)}$    & 0.0362 & 0.0537 & 0.0404 & 0.0178 &  \textbf{0.0147}   \\
        $\sigma = 120^{(\blacktriangle)}$   & 0.0402 & 0.0569 & 0.0489 & 0.0241 &  \textbf{0.0237}    \\
        $\sigma = 150^{(\circ)}$   & 0.0543 & 0.0696 & 0.0628 & 0.0494 &  \textbf{0.0318}   \\
        $\sigma = 200^{(\bullet)}$    & 0.9556 & 1.1124 & 1.0376 & \textbf{0.6841} &  0.8042   \\
        \bottomrule
    \end{tabular}
    \label{KoloTable}
\end{table}

        

Table \ref{KoloTable} shows the performance of different models in comparison. 
Notations remain the same as in table \ref{laminarTable}. It is seen that DP and BDP methods perform much better than diffusion model (DM) and its recent one-step adaptations. In Figure \ref{Ref3DKolo100} to \ref{DMShort3DKolo100}, we plot the distribution generated by DP models and single-step DM models. By comparing with the reference generated by the IPM, we observe that the distributions learned by DP models come much closer to the target.

\begin{figure}[H]
    \centering
    \begin{subfigure}{0.31\textwidth}
         \includegraphics[width = \linewidth, height = 4cm]{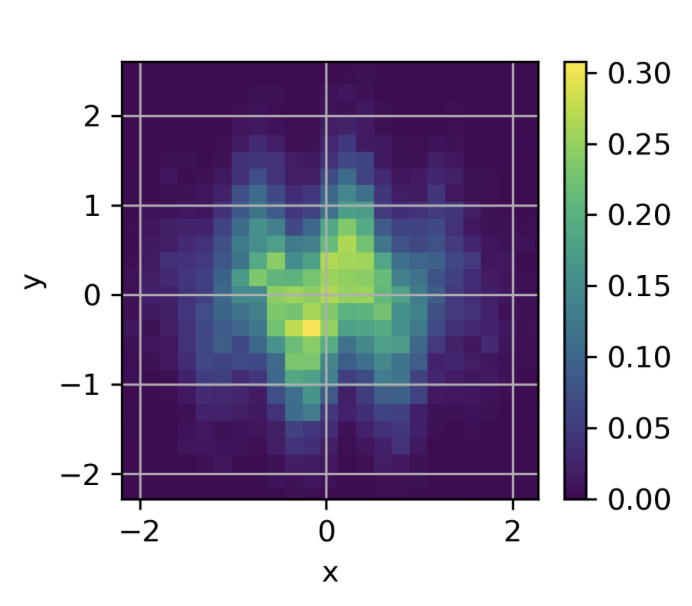}   
    \end{subfigure}
    \begin{subfigure}{0.31\textwidth}
         \includegraphics[width = \linewidth, height = 4cm]{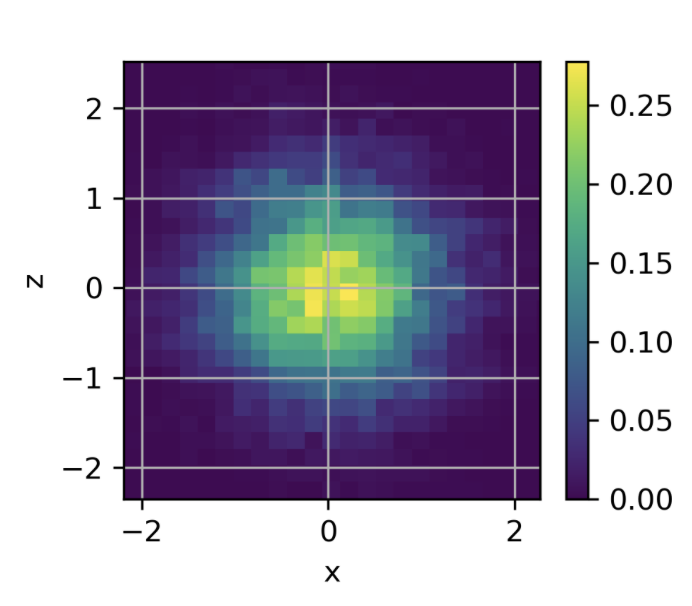}    
    \end{subfigure}
    \begin{subfigure}{0.31\textwidth}
         \includegraphics[width = \linewidth, height = 4cm]{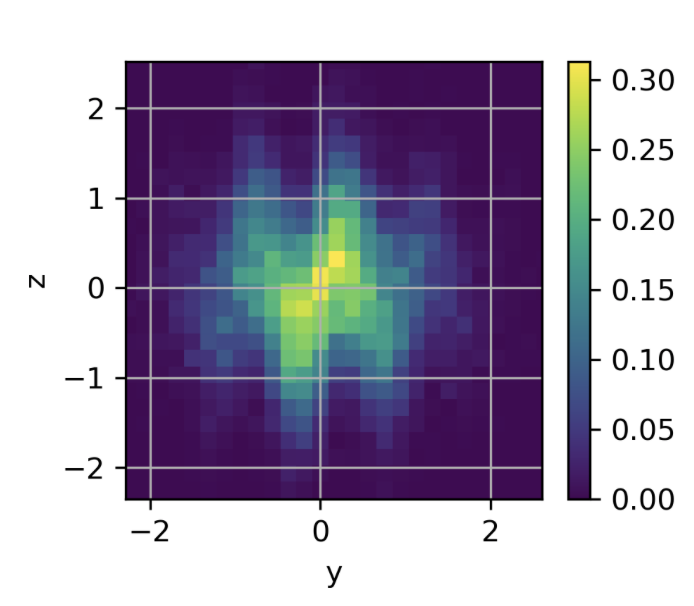}   
    \end{subfigure}
    \caption{Cross sectional views of  reference in 3D Kolmogorov flow at $\sigma = 100$.}
    \label{Ref3DKolo100}
\end{figure}

\begin{figure}[H]
    \centering
    \begin{subfigure}{0.31\textwidth}
         \includegraphics[width = \linewidth, height = 4cm]{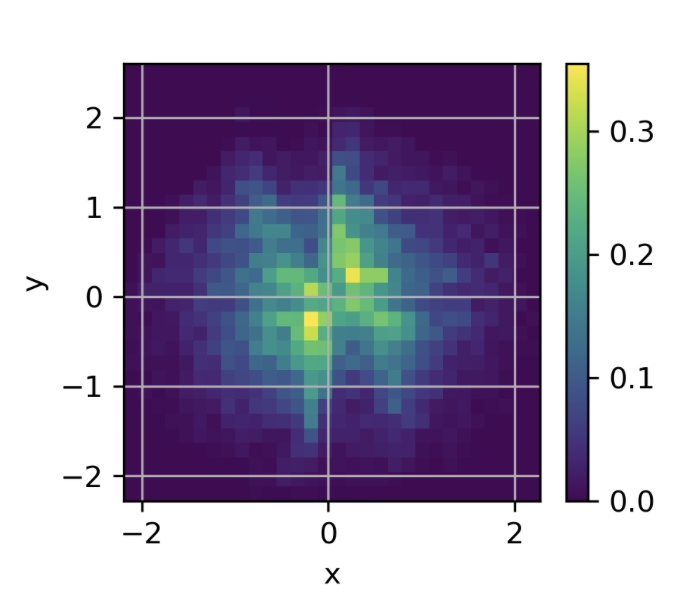}   
    \end{subfigure}
    \begin{subfigure}{0.31\textwidth}
         \includegraphics[width = \linewidth, height = 4cm]{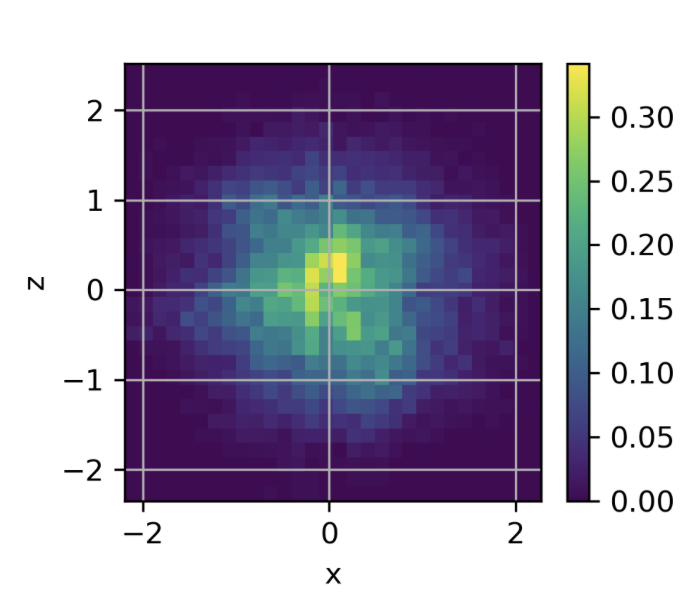}   
    \end{subfigure}
    \begin{subfigure}{0.31\textwidth}
         \includegraphics[width = \linewidth, height = 4cm]{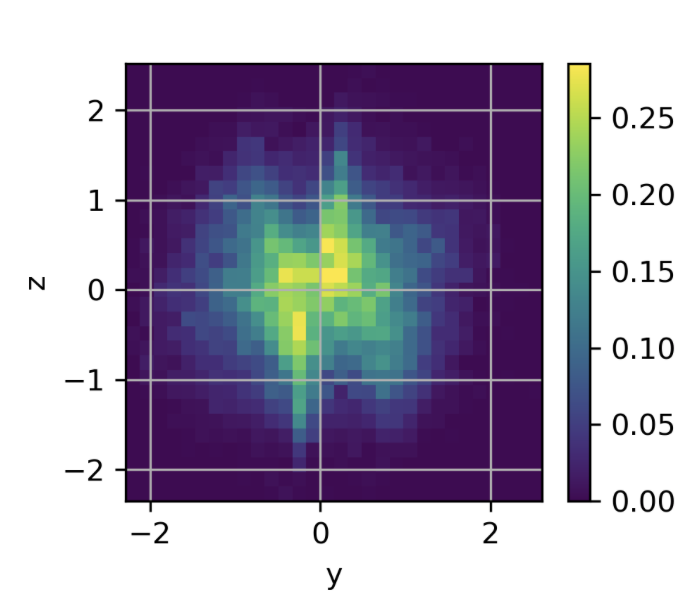}   
    \end{subfigure}
    \caption{Cross sectional views of  DP method in 3D Kolmogorov flow at $\sigma = 100$.}
    \label{DP3DKolo100}
\end{figure}

\begin{figure}[H]
    \centering
    \begin{subfigure}{0.31\textwidth}
         \includegraphics[width = \linewidth, height = 4cm]{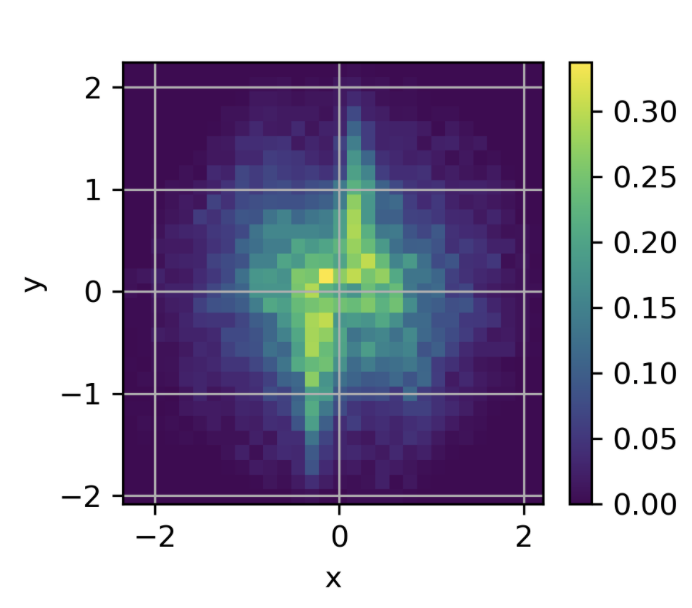}   
    \end{subfigure}
    \begin{subfigure}{0.31\textwidth}
         \includegraphics[width = \linewidth, height = 4cm]{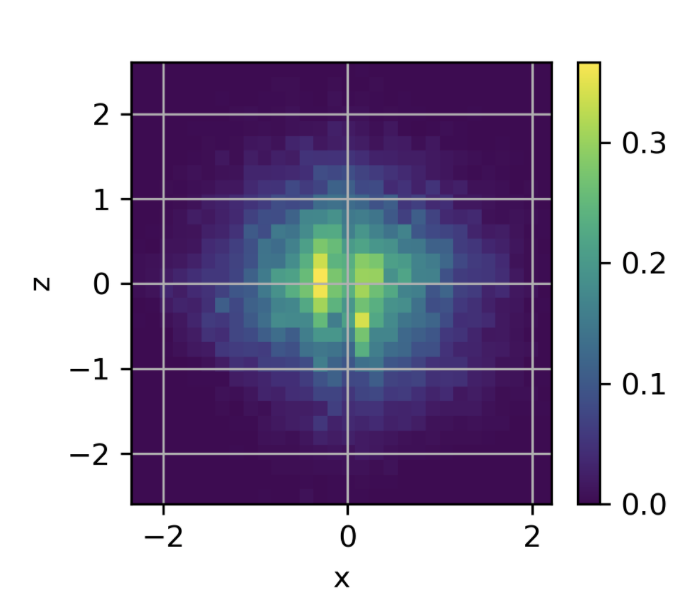}   
    \end{subfigure}
    \begin{subfigure}{0.31\textwidth}
         \includegraphics[width = \linewidth, height = 4cm]{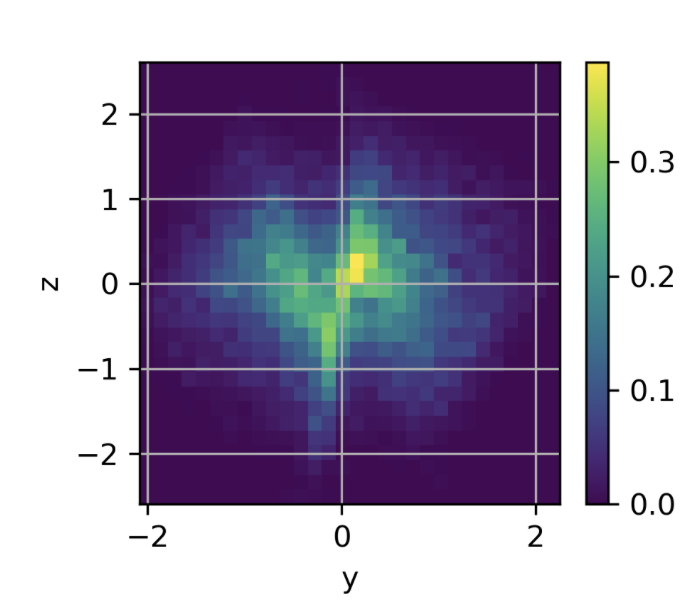}   
    \end{subfigure}
    \caption{Cross sectional views of  BDP method in 3D Kolmogorov flow at $\sigma = 100$.}
    \label{DPBir3DKolo100}
\end{figure}

\begin{figure}[H]
    \centering
    \begin{subfigure}{0.31\textwidth}
         \includegraphics[width = \linewidth, height = 4cm]{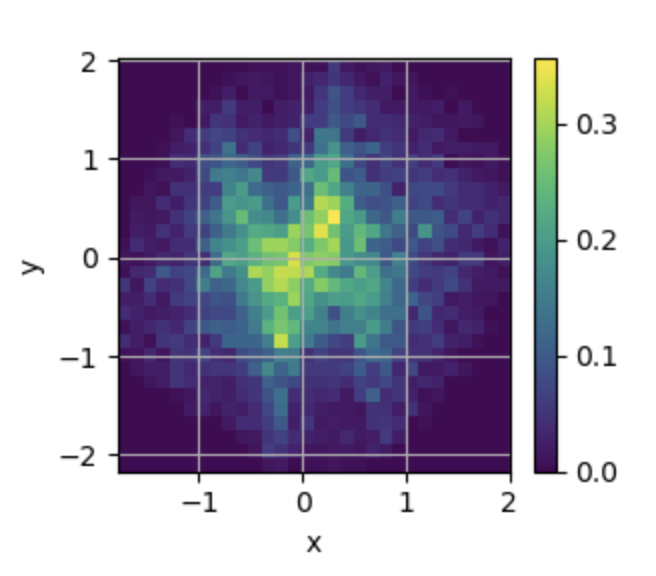}   
    \end{subfigure}
    \begin{subfigure}{0.31\textwidth}
         \includegraphics[width = \linewidth, height = 4cm]{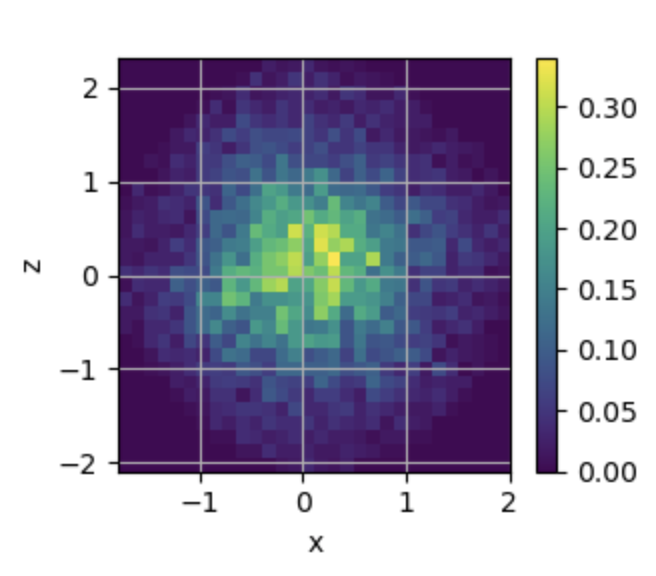}   
    \end{subfigure}
    \begin{subfigure}{0.31\textwidth}
         \includegraphics[width = \linewidth, height = 4cm]{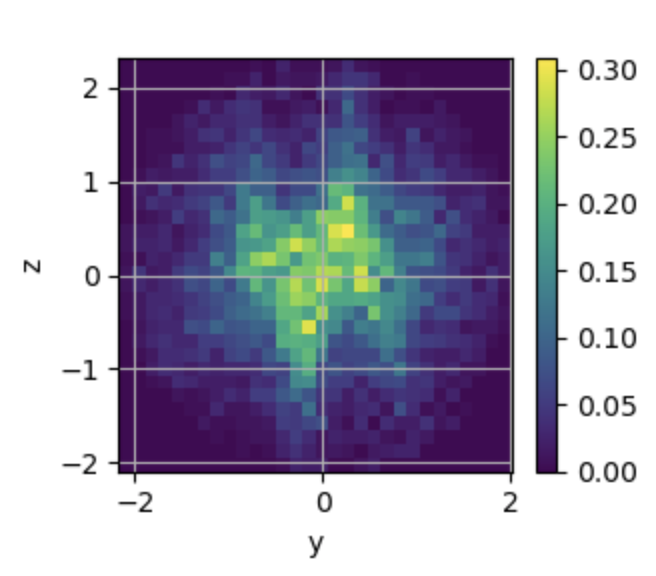}   
    \end{subfigure}
    \caption{Cross sectional views of  Rectified flow in 3D Kolmogorov flow at $\sigma = 100$.}
    \label{DMRect3DKolo100}
\end{figure}

\begin{figure}[H]
    \centering
    \begin{subfigure}{0.31\textwidth}
         \includegraphics[width = \linewidth, height = 4cm]{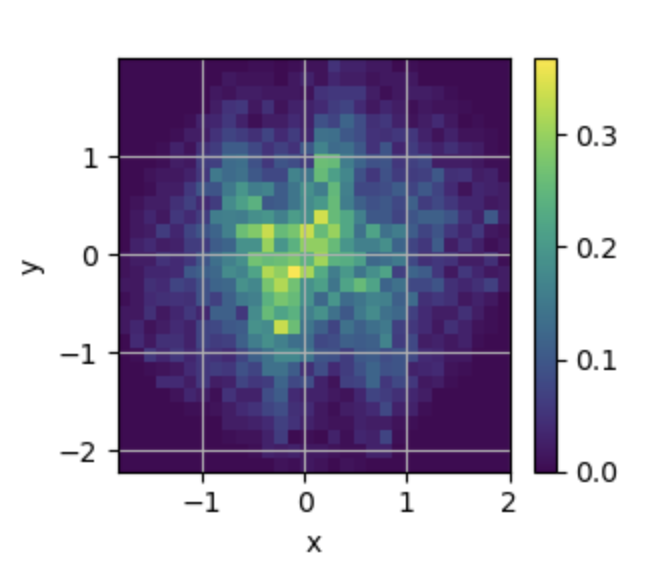}   
    \end{subfigure}
    \begin{subfigure}{0.31\textwidth}
         \includegraphics[width = \linewidth, height = 4cm]{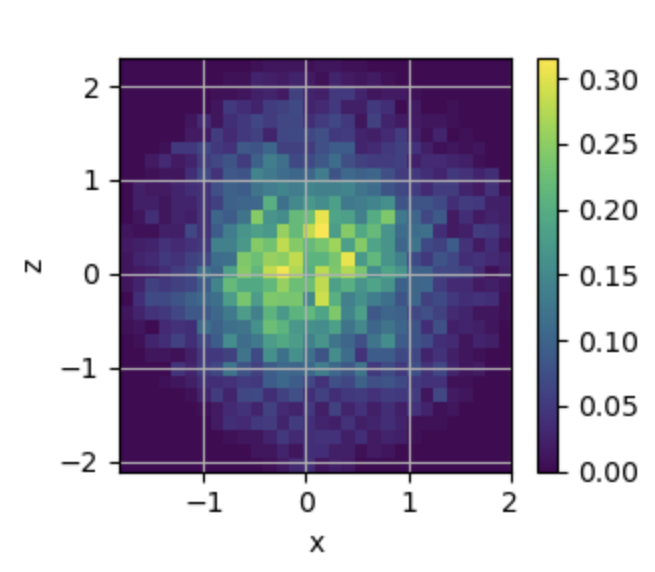}   
    \end{subfigure}
    \begin{subfigure}{0.31\textwidth}
         \includegraphics[width = \linewidth, height = 4cm]{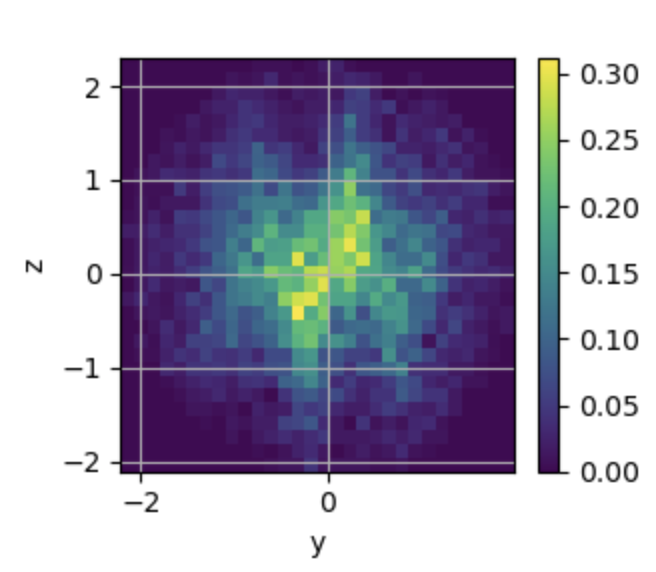}   
    \end{subfigure}
    \caption{Cross sectional views of  Shortcut model in 3D Kolmogorov flow at $\sigma = 100$.}
    \label{DMShort3DKolo100}
\end{figure}

\subsection{Mixtures of Gaussian}
In this subsection, we compare the performance of BDP and single-step diffusion (flow matching) models in the case of Gaussian mixtures in different space dimensions. For a mixture of $m$ Gaussian distributions on $\mathbb{R}^n$, the target distribution is in closed form. Let the probability density function of a multivariate Gaussian distribution with mean $\bm \mu$ and covariance matrix $\bm \Sigma$ be:
\begin{equation}
    p(\bm x; \bm \mu, \bm \Sigma) = \frac{1}{(2\pi)^{n/2} |\bm \Sigma|^{1/2}} \exp (-\frac{1}{2} (\bm x - \bm  \mu)^T \bm \Sigma^{-1} (\bm x - \bm \mu)).
\end{equation}
Then the close form of the Gaussian mixture can be represented by 
\begin{equation}
    p_{GM}(\bm x; \{\bm \mu_i, \bm \Sigma_i\}_{i=1}^m  = \sum_{i=1}^m w_i\cdot p(\bm x; \bm \mu_i, \bm \Sigma_i),\ with\ \sum_{i=1}^m w_i = 1. 
\end{equation}
Here $w_i > 0$ is the weight of the $i$-th Gaussian distribution. In the numerical experiments, we choose $m = 2$ and set $w_i = 0.5, i = 1,2$.  The following table shows the performance of DP models, diffusion and two single-step diffusion models vs. dimensions. The \textit{2-Wasserstein} distance is computed between the first two dimensions of the network output and the reference.

\begin{table}[h]
    \centering
    \caption{Model performance in mixture of two Gaussians on $\mathbb{R}^{n}$, bold/blue is best among all (one-step) methods. }
    \begin{tabular}{cccccc}
        \toprule
        \multirow{3}{*}{$n\ \setminus$ Model} & 
        & \multicolumn{3}{c}{$W_2$ distance (parameter size $3k$)} \\
        \cline{2-6}
           & DM & DM & DM & DP & DP\\
        \cline{2-6}
           &   & Rectified flow & Shortcut  & & Bi-direction\\
        
        \midrule
        $n = 3$  &  0.0633 & 0.0812 & 0.0654 & 0.0635 & \textbf{0.0534}  \\
        $n = 4$    & \textbf{0.0603} & 0.0840 & \textcolor{blue}{0.0632} & 0.1141 &  0.1052    \\
        $n = 8$    & \textbf{0.0657} & 0.0826 & \textcolor{blue}{0.0688} & 0.1317 &  0.1207    \\
        $n = 16$   & \textbf{0.0700} & 0.0896 & \textcolor{blue}{0.0739} & 0.1598 &  0.1493    \\
        $n = 32$   & \textbf{0.1035} & 0.1369 & \textcolor{blue}{0.1124} & 0.3507 &  0.2926    \\
        
        \bottomrule
    \end{tabular}
    \label{GMMTable}
\end{table}

From Table \ref{GMMTable}, we observe that the two single-step models maintain performance better than DP as the dimension increases. 
This is due to the transition matrix computations in DP models which have $O(N^2)$ memory cost, $N$ the mini-batch size. However, DM and flow matching (FM) models only need $O(N)$ memory for the cost related to batch size. As the dimension increases, the DP-based models have more information to learn in each round of learning, but the $O(N^2)$ memory cost growth and the memory limitation of the machine itself will make it challenging to provide more data samples in one epoch. In comparison, DM-based models are much less sensitive in this regard. In lower dimensions, such as $2$-$3$ space dimensions, or when data volume is not large, the DP models can achieve better performance while remaining efficient, which is consistent with the above two examples of KS models.

\begin{figure}[h]
    \centering
    \begin{subfigure}{0.325\textwidth}
         \includegraphics[width = \linewidth, height = 3.97cm]{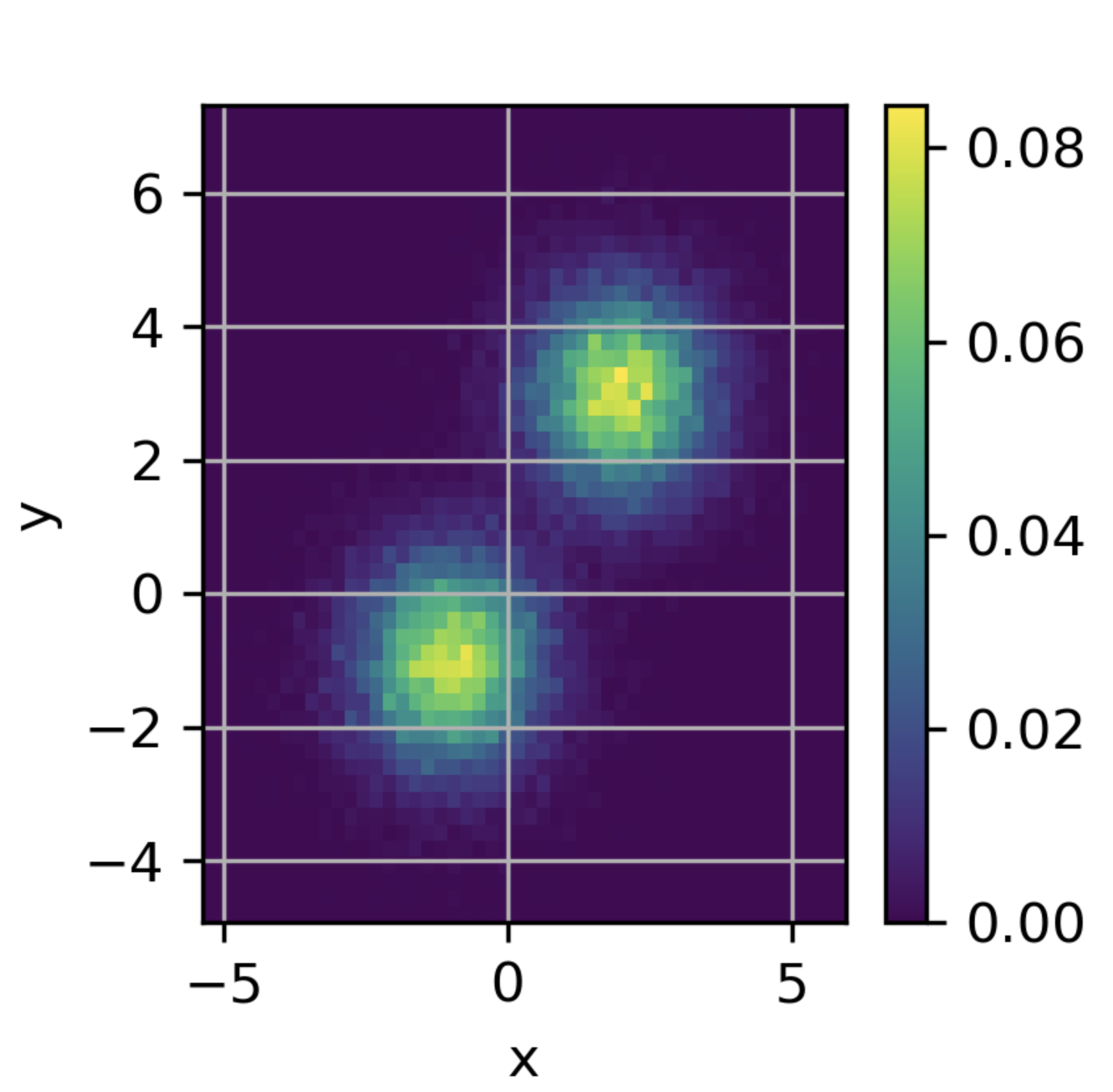}
         \caption{Reference}
         \label{16DDMcutDiffstep1k}
    \end{subfigure}
    \begin{subfigure}{0.325\textwidth}
         \includegraphics[width = \linewidth, height = 4cm]{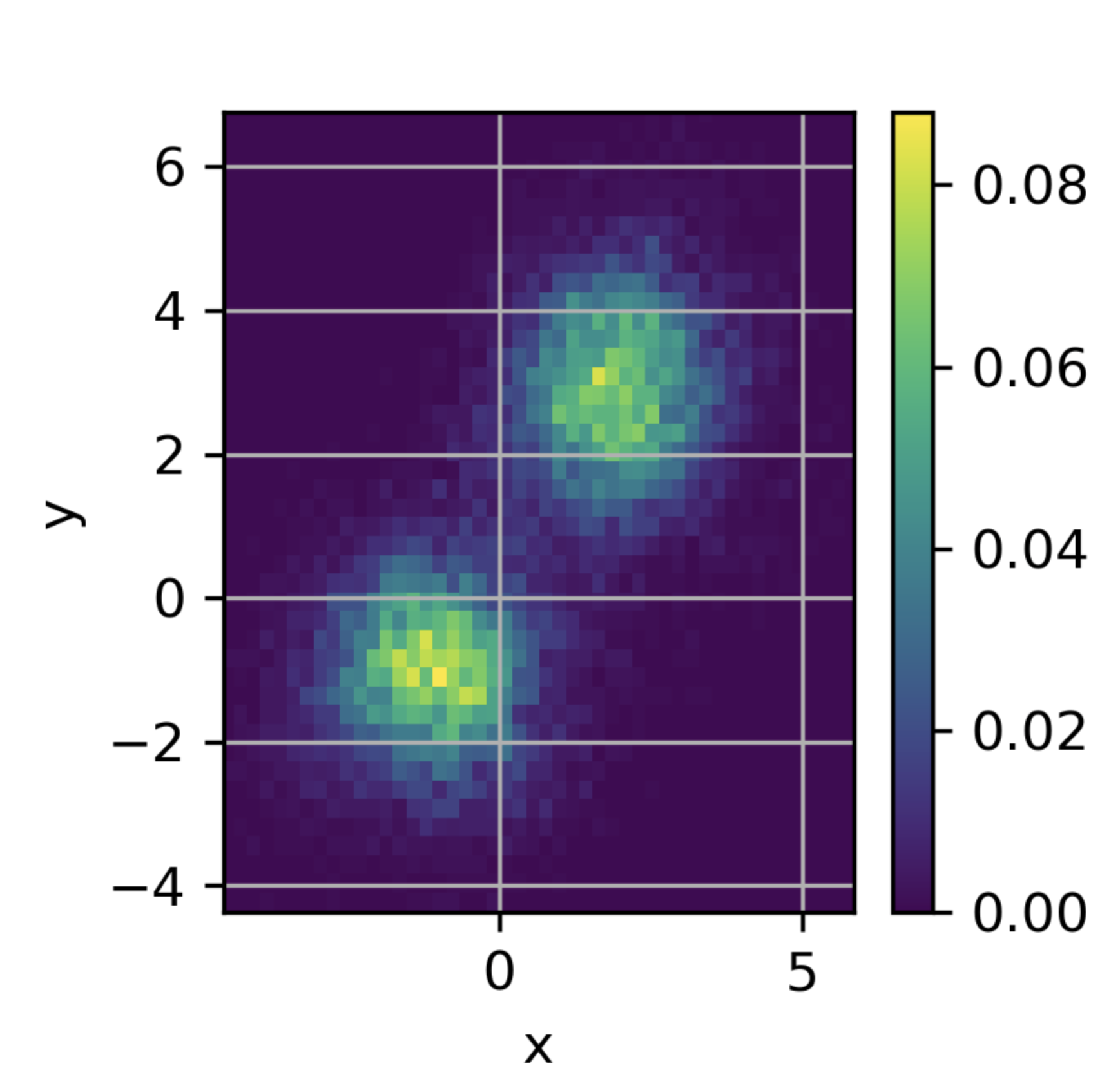}
         \caption{Deep Particle}
         \label{16DDMDP}
    \end{subfigure}
    \begin{subfigure}{0.325\textwidth}
         \includegraphics[width = \linewidth, height = 4cm]{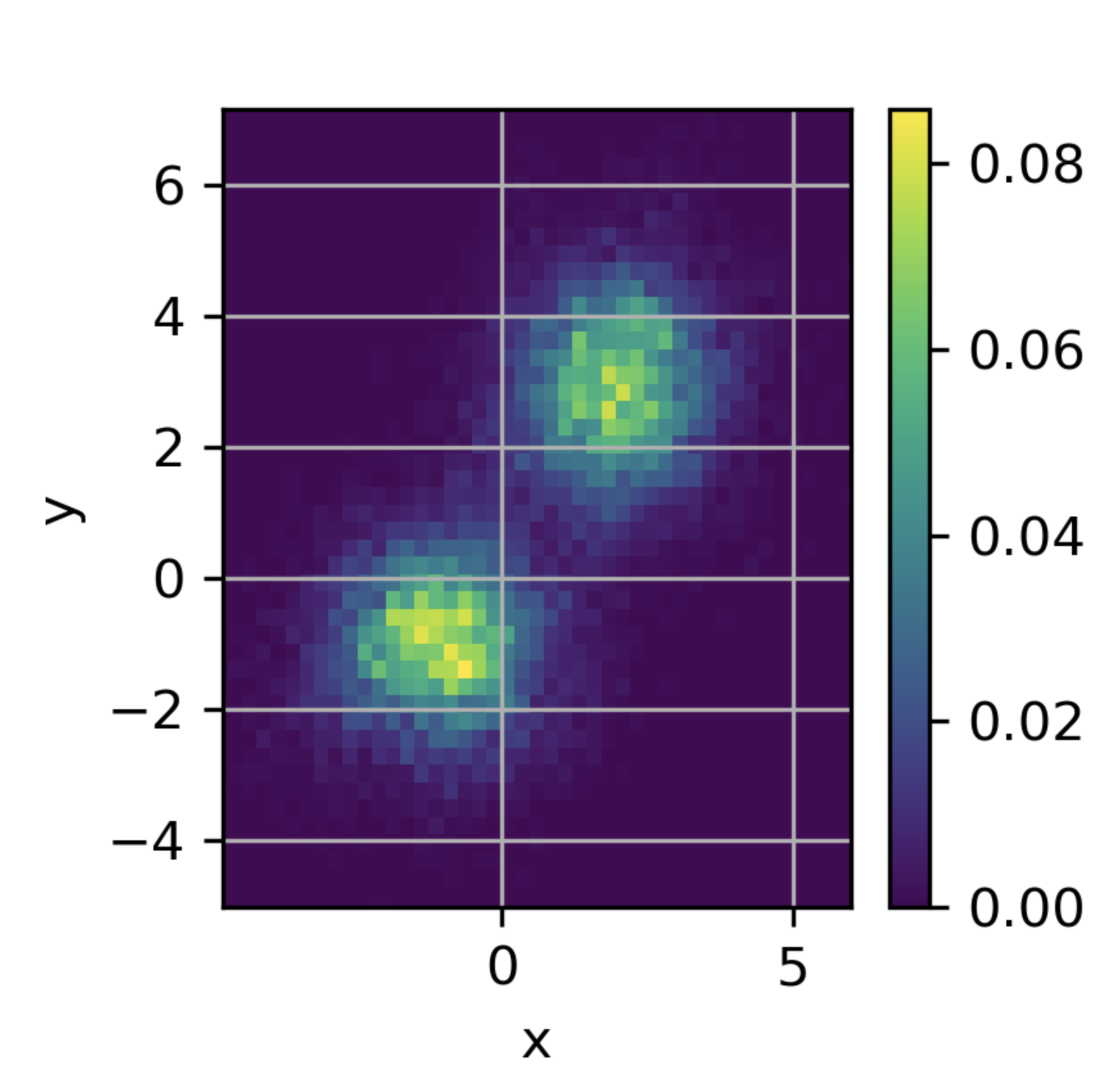}
         \caption{BDP Method}
         \label{16DDPBir}
    \end{subfigure}

    \begin{subfigure}{0.325\textwidth}
         \includegraphics[width = \linewidth, height = 3.97cm]{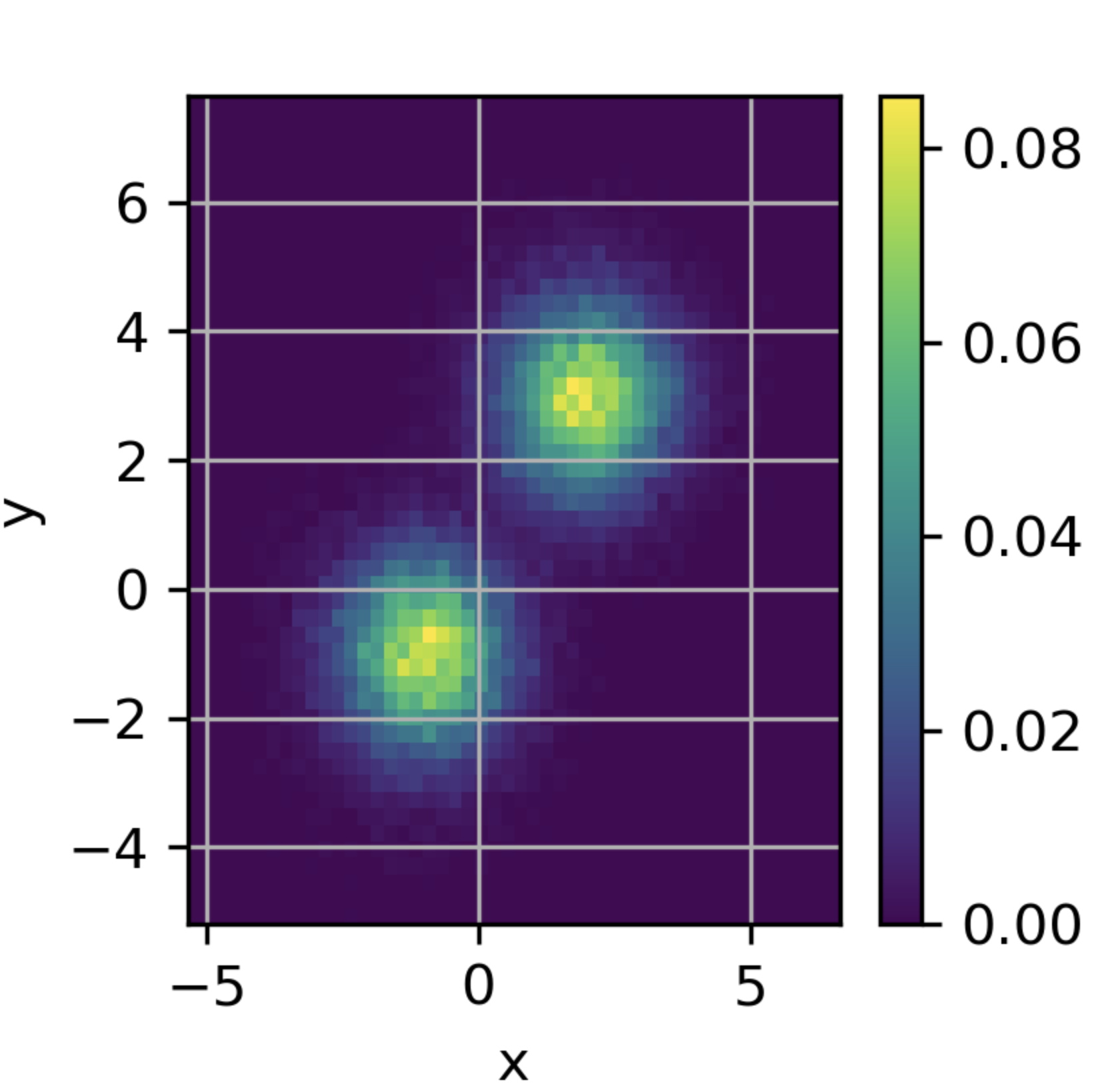}
         \caption{Diffusion Model with step $10^3$}
         \label{16DDMcutDiffstep100}
    \end{subfigure}
    \begin{subfigure}{0.325\textwidth}
         \includegraphics[width = \linewidth, height = 4cm]{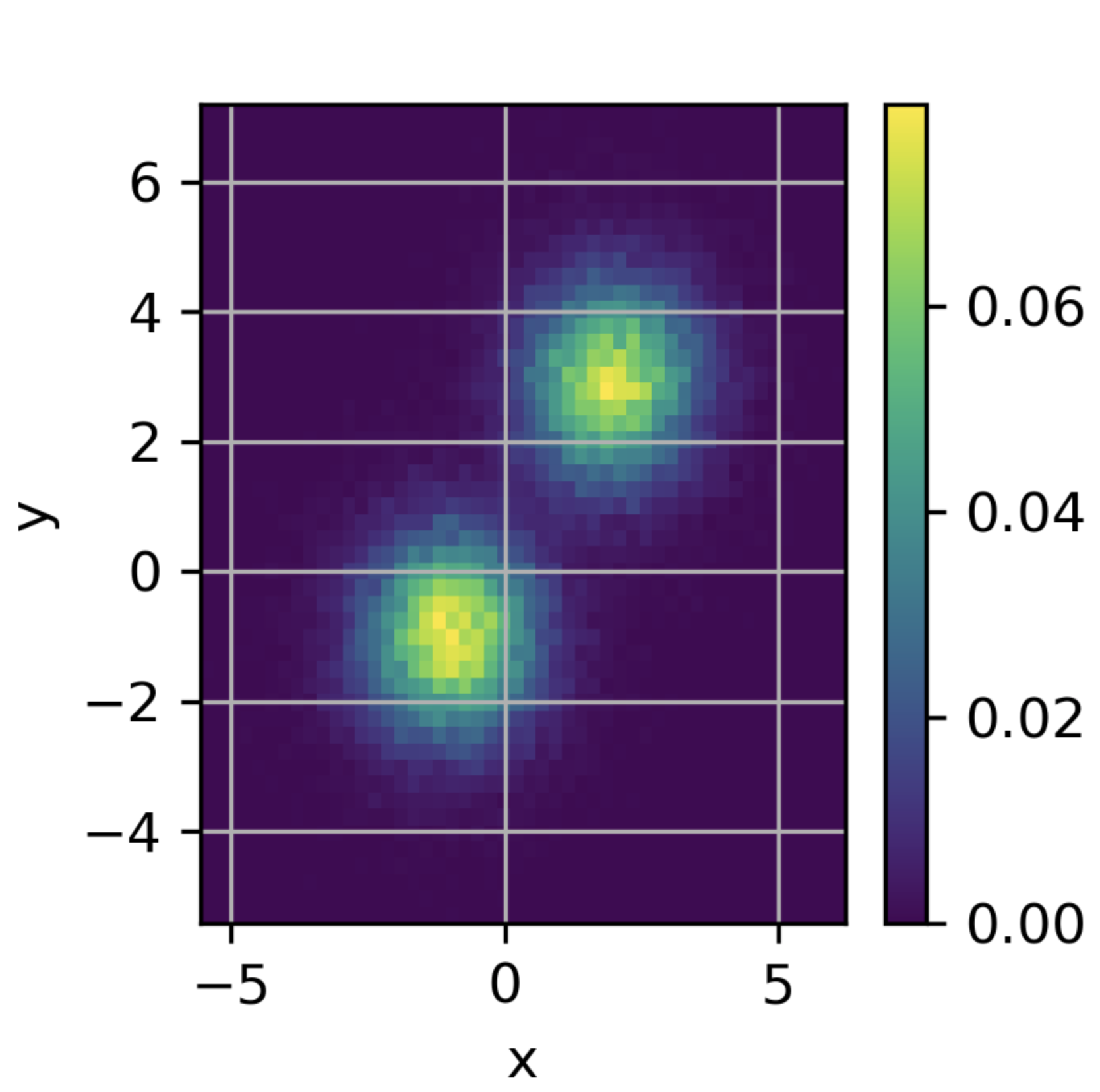}
         \caption{Rectified flow}
         \label{16DDMRectdt1}
    \end{subfigure}
    \begin{subfigure}{0.325\textwidth}
         \includegraphics[width = \linewidth, height = 4cm]{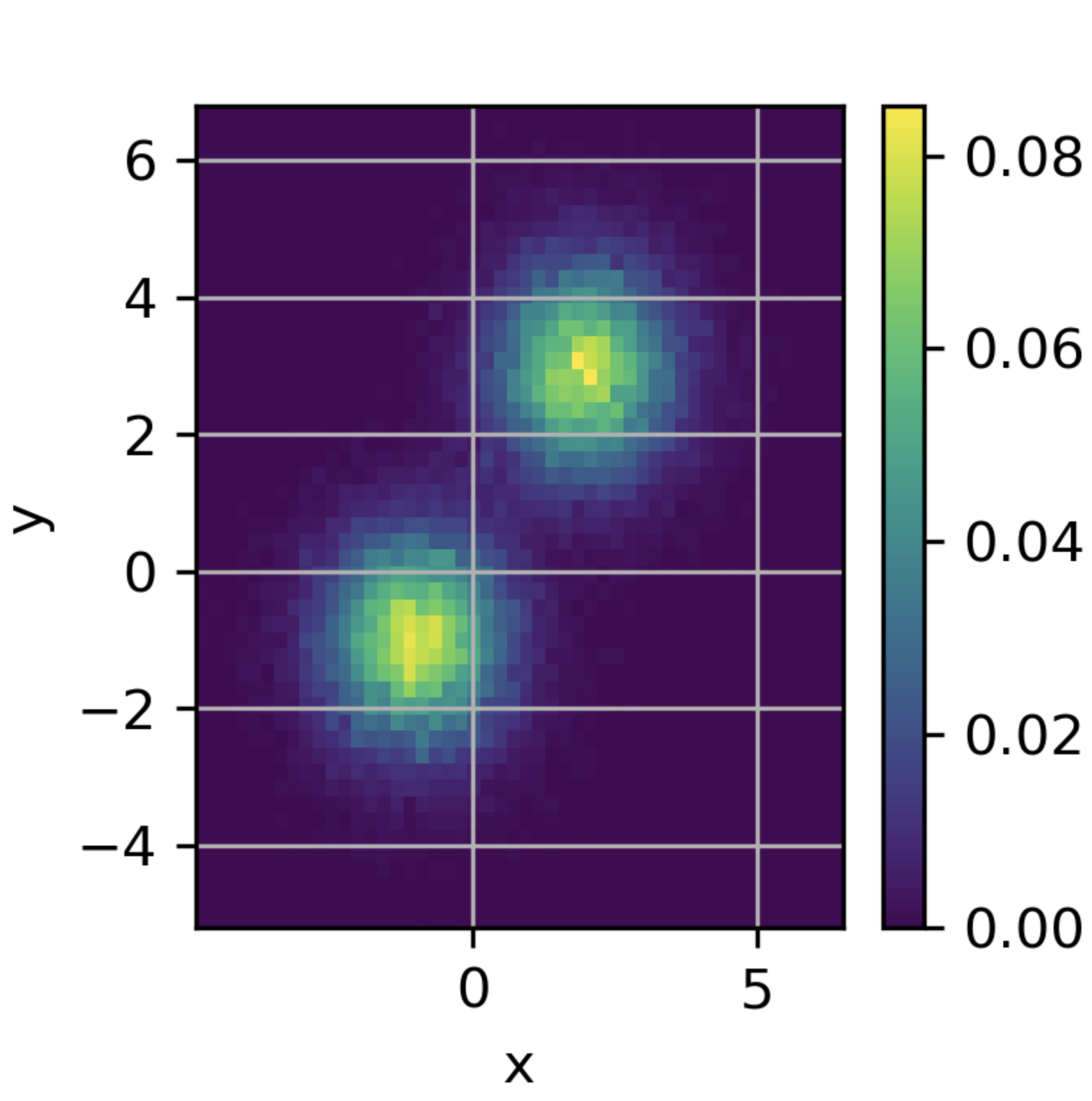}
         \caption{ShortCut Model}
         \label{16DDMShortCut}
    \end{subfigure}

    \caption{Generation projected on the $x$-$y$ plane by models vs. reference, targeting a $16 D$ Gaussian mixture.}
\end{figure}

\section{Conclusion}\label{sec:Conclusion}
\noindent We developed a deep learning approach, the BDP method, to efficiently solve a class of physically parameterized transport map problems in low (physical) dimensions. The approach is based on the computation of \textit{2-Wasserstein distance} and simultaneous learning of the forward and reverse mappings. The bidirectional architecture improves the stability of the learned mapping. We adopted the \textit{mini-batch} technique in the training process and analyzed the approximation error in terms of mini-batch size.
We then studied the performance of the BDP method and compared it with two single-step diffusion models on the Keller-Segel system in the presence of 3D laminar and Kolmogorov flows, as well as the mixture of Gaussians. The BDP method performs better in physical dimensions when the parameter size is moderate, while the one-step diffusion (flow matching) type generative models scale better in higher dimensions and large volumes of data. The inference speed of BDP remains the fastest in all dimensions observed.

In future work, we plan to further study the critical phenomenon (from DP to diffusion models) in terms of transition dimension and develop methods to improve the scalability of DP.

\section*{Acknowledgements}
\noindent ZW was partly supported by NTU SUG-023162-00001, MOE AcRF Tier 1 Grant RG17/24. JX was partly supported by NSF grants DMS-2309520, DMS-2219904, DMS-2151235.  ZZ was partially supported by the National Natural Science Foundation of China (Projects 92470103 and 12171406), the Hong Kong RGC grant (Projects 17307921 and 17304324), the Seed Funding Programme for Basic Research (HKU), the Outstanding Young Researcher Award of HKU (2020-21), and Seed Funding for Strategic Interdisciplinary Research Scheme 2021/22 (HKU). The computations were performed using research computing facilities provided by Information Technology Services, The University of Hong Kong.



\bibliographystyle{plain}
\bibliography{reference}

\begin{thebibliography}{10}

\bibitem{altschuler2017near}
Jason Altschuler, Jonathan Niles-Weed, and Philippe Rigollet.
\newblock Near-linear time approximation algorithms for optimal transport via {Sinkhorn} iteration.
\newblock {\em Advances in neural information processing systems}, 30, 2017.

\bibitem{ambrosio2021lectures}
Luigi Ambrosio, Elia Bru{\'e}, Daniele Semola, et~al.
\newblock {\em Lectures on optimal transport}, volume 130.
\newblock Springer, 2021.

\bibitem{chen2018neural}
Ricky~TQ Chen, Yulia Rubanova, Jesse Bettencourt, and David~K Duvenaud.
\newblock Neural ordinary differential equations.
\newblock {\em Advances in neural information processing systems}, 31, 2018.

\bibitem{courty2016optimal}
Nicolas Courty, R{\'e}mi Flamary, Devis Tuia, and Alain Rakotomamonjy.
\newblock Optimal transport for domain adaptation.
\newblock {\em IEEE transactions on pattern analysis and machine intelligence}, 39(9):1853--1865, 2016.

\bibitem{dinh2017density}
Laurent Dinh, Jascha Sohl-Dickstein, and Samy Bengio.
\newblock Density estimation using real nvp.
\newblock In {\em International Conference on Learning Representations}, 2017.

\bibitem{fatras2021minibatch}
Kilian Fatras, Younes Zine, Szymon Majewski, R{\'e}mi Flamary, R{\'e}mi Gribonval, and Nicolas Courty.
\newblock Minibatch optimal transport distances; analysis and applications.
\newblock {\em arXiv preprint arXiv:2101.01792}, 2021.

\bibitem{figalli2021invitation}
Alessio Figalli and Federico Glaudo.
\newblock {\em An invitation to optimal transport, Wasserstein distances, and gradient flows}.
\newblock 2021.

\bibitem{frans2024one}
Kevin Frans, Danijar Hafner, Sergey Levine, and Pieter Abbeel.
\newblock One step diffusion via shortcut models.
\newblock {\em arXiv preprint arXiv:2410.12557, in International Conference on Learning Representations}, 2025.

\bibitem{GAN}
Ian~J. Goodfellow, Jean Pouget-Abadie, Mehdi Mirza, Bing Xu, David Warde-Farley, Sherjil Ozair, Aaron Courville, and Yoshua Bengio.
\newblock Generative adversarial nets.
\newblock In {\em Proceedings of the 28th International Conference on Neural Information Processing Systems - Volume 2}, page 2672–2680. MIT Press, 2014.

\bibitem{ho2020denoising}
Jonathan Ho, Ajay Jain, and Pieter Abbeel.
\newblock Denoising diffusion probabilistic models.
\newblock {\em Advances in neural information processing systems}, 33:6840--6851, 2020.

\bibitem{keller1970initiation}
Evelyn~F Keller and Lee~A Segel.
\newblock Initiation of slime mold aggregation viewed as an instability.
\newblock {\em Journal of theoretical biology}, 26(3):399--415, 1970.

\bibitem{kingma2013auto}
Diederik~P. Kingma and Max Welling.
\newblock Auto-encoding variational bayes.
\newblock In {\em International Conference on Learning Representations}, 2013.

\bibitem{liu2022flow}
Xingchao Liu, Chengyue Gong, et~al.
\newblock Flow straight and fast: Learning to generate and transfer data with rectified flow.
\newblock In {\em International Conference on Learning Representations}, 2023.

\bibitem{papamakarios2021normalizing}
George Papamakarios, Eric Nalisnick, Danilo~Jimenez Rezende, Shakir Mohamed, and Balaji Lakshminarayanan.
\newblock Normalizing flows for probabilistic modeling and inference.
\newblock {\em Journal of Machine Learning Research}, 22(57):1--64, 2021.

\bibitem{peyre2019computational}
Gabriel Peyr{\'e}, Marco Cuturi, et~al.
\newblock Computational optimal transport: With applications to data science.
\newblock {\em Foundations and Trends{\textregistered} in Machine Learning}, 11(5-6):355--607, 2019.

\bibitem{sinkhorn1964relationship}
Richard Sinkhorn.
\newblock A relationship between arbitrary positive matrices and doubly stochastic matrices.
\newblock {\em The annals of mathematical statistics}, 35(2):876--879, 1964.

\bibitem{sommerfeld2019optimal}
Max Sommerfeld, J{\"o}rn Schrieber, Yoav Zemel, and Axel Munk.
\newblock Optimal transport: Fast probabilistic approximation with exact solvers.
\newblock {\em Journal of Machine Learning Research}, 20(105):1--23, 2019.

\bibitem{song2020denoising}
Jiaming Song, Chenlin Meng, and Stefano Ermon.
\newblock Denoising diffusion implicit models.
\newblock In {\em International Conference on Learning Representations}, 2021.

\bibitem{song2020score}
Yang Song, Jascha Sohl-Dickstein, Diederik~P Kingma, Abhishek Kumar, Stefano Ermon, and Ben Poole.
\newblock Score-based generative modeling through stochastic differential equations.
\newblock In {\em International Conference on Learning Representations}, 2021.

\bibitem{trigila2016data}
Giulio Trigila and Esteban~G Tabak.
\newblock Data-driven optimal transport.
\newblock {\em Communications on Pure and Applied Mathematics}, 69(4):613--648, 2016.

\bibitem{villani2021topics}
C{\'e}dric Villani.
\newblock {\em Topics in optimal transportation}, volume~58.
\newblock American Mathematical Soc., 2021.

\bibitem{DP_22}
Zhongjian Wang, Jack Xin, and Zhiwen Zhang.
\newblock Deep{P}article: learning invariant measure by a deep neural network minimizing {W}asserstein distance on data generated by an interacting particle method.
\newblock {\em Journal of Computational Physics}, 464:111309, 2022.

\bibitem{wang2024deepparticle}
Zhongjian Wang, Jack Xin, and Zhiwen Zhang.
\newblock A {DeepParticle} method for learning and generating aggregation patterns in multi-dimensional {Keller-Segel} chemotaxis systems.
\newblock {\em Physica D: Nonlinear Phenomena}, 460:134082, 2024.

\bibitem{zhu2017unpaired}
Jun-Yan Zhu, Taesung Park, Phillip Isola, and Alexei~A Efros.
\newblock Unpaired image-to-image translation using cycle-consistent adversarial networks.
\newblock In {\em Proceedings of the IEEE international conference on computer vision}, pages 2223--2232, 2017.

\end{thebibliography}

\end{document}